\documentclass[a4paper]{ar-1col}

\usepackage[comma,authoryear]{natbib}
\setcitestyle{authoryear}
\usepackage{url}
\usepackage{amssymb}
\usepackage{amsmath}
\usepackage{amsthm}
\usepackage{mathtools}

\newtheorem{theorem}{Theorem}

\jname{Xxxx. Xxx. Xxx. Xxx.}
\jvol{AA}
\jyear{YYYY}
\doi{10.1146/((please add article doi))}

\newcommand{\flag}[2]{\genfrac{[}{]}{0pt}{}{#1}{#2}}

\begin{document}

\markboth{Carr\'on Duque J, Marinucci D}{Geometric Methods for Spherical Data, with Applications to Cosmology}

\title{Geometric Methods for Spherical Data, with Applications to Cosmology}

\author{Javier Carr\'on Duque$^{1,2}$ and Domenico Marinucci$^3$
\affil{$^1$Department of Physics, University of Rome Tor Vergata, Roma, Italy, 00133; email: javier.carron@roma2.infn.it}
\affil{$^2$Sezione INFN Roma~2, Roma, Italy, 00133}
\affil{$^3$Department of Mathematics, University of Rome Tor Vergata, Roma, Italy, 00133; email: marinucc@mat.uniroma2.it }}

\begin{abstract}

This survey is devoted to recent developments in the statistical analysis of
spherical data, with a view to applications in Cosmology. We will start from
a brief discussion of Cosmological questions and motivations, arguing that most Cosmological observables are spherical random fields. Then, we will introduce some mathematical background on spherical random fields, including
spectral representations and the construction of needlet and wavelet frames. We
will then focus on some specific issues, including tools and algorithms for map reconstruction (\textit{i.e.}, separating the different physical components which contribute to the observed field), geometric tools for testing the assumptions of Gaussianity and isotropy, and multiple testing methods to detect contamination in the field due to point sources. Although these tools are introduced in the Cosmological context, they can be applied to other situations dealing with spherical data. Finally, we will discuss more recent and challenging issues such as the analysis of polarization data, which can be viewed as realizations of random fields taking values in spin fiber bundles.
\end{abstract}

\maketitle

\tableofcontents

\section{BACKGROUND: COSMOLOGICAL MOTIVATIONS}

The status of Cosmology as an observational Science has experienced a
dramatic shift in the last two decades. Until about the year 2000, Cosmology
was known as a data-starved science: important experiments had been carried
over in a number of fields, but the amount and quality of data were very far
from the level needed to produce precise estimates of parameters and
conclusive tests of most Cosmological Models; the situation has changed
completely in the last couple of decades, when a number of experiments have improved the size and precision of existing observations by several orders of magnitude.

To name only a few such experiments, the Cosmic Microwave Background (CMB; to be
discussed at length below) has been probed by the satellite experiments WMAP
and Planck, by several balloon-borne and by ground-based observatories, and it
is going to be further investigated by the forthcoming satellite mission
LiteBIRD; Gamma rays sources have been probed by Satellite Missions
Fermi-LAT and AGILE, and from several observatories on the ground;
ultra-high energy cosmic rays are investigated by huge international
collaborations such as the Pierre Auger observatory; Cosmic Neutrinos are the
object of investigation by IceCube; Gravitational Waves have been detected
by the LIGO--Virgo Collaboration and will be further probed in the next decades by a new
generation of observatories; radio sources are probed by huge
collaborations such as SKA, whereas the Large Scale Structure of the
Universe and weak gravitational lensing are the object of upcoming missions
such as Euclid.

A remarkable feature of all the datasets produced by these observations is the
following: they are all defined on the sphere $\mathbb{S}^{2}.$ For reasons
of space and clarity, we will concentrate most of the discussion below on the
CMB, but the tools that we shall introduce for
spherical data analysis are relevant for many of the other applications as
well.

To understand the CMB radiation, let us recall the Standard Cosmological Model:
it states that the Universe (or the Universe that we
currently observe) started 13.7 Billion years ago in a very hot and
dense state, filled by a plasma of electrons, photons and protons
(we note that this is a large oversimplification, but it does fit our purposes,
see, \textit{e.g.}, \cite{dode2004}, \cite{Durrer}, \cite{Vittorio}, for a more detailed description).
Matter was completely ionized, meaning that the mean kinetic energy of the electrons
was higher than the electromagnetic potential of the protons, and no stable
atoms formed: free electrons have a much larger "scattering surface", \textit{i.e.},
probability to interact with a photon, so that the mean free path of the
latter was very short and the Universe was basically opaque to light. As the
Universe expanded, the energy density and the kinetic energy of the
electrons decreased to a point where it was no longer enough to resist the
attraction of the protons: at this stage, stable hydrogen atoms formed, and the Universe became transparent to light. In this so--called ``Age of Recombination'',
which is now estimated to have taken place about 377,000 years after the Big
Bang, the Universe became transparent to light and these primordial photons
started to move freely across the Universe. Hence, one of the key predictions of the
model is that we should live embedded into this relic radiation, providing an image of the Universe as
it was 13.7 billion years ago. These photons, now observed in microwave frequencies and in every direction on the sky, constitute the Cosmic Microwave Background.

Although the first papers predicting the CMB radiation date back to around 1945,
the first observational proof of its existence was given in a celebrated
experiment by Penzias and Wilson in 1965. However, it was only in the current century that sophisticated satellite missions such as WMAP and, especially, Planck managed to produce high-resolution low-noise full-sky maps of the \emph{Last Scattering
Surface}, see \textbf{Figure \ref{f:planck}}, \cite{planck}, and \cite{starck2014}. In the next Section, we shall introduce the mathematical formalism that we will require for the statistical analysis of
these maps.

\begin{figure}
\includegraphics[width=5in]{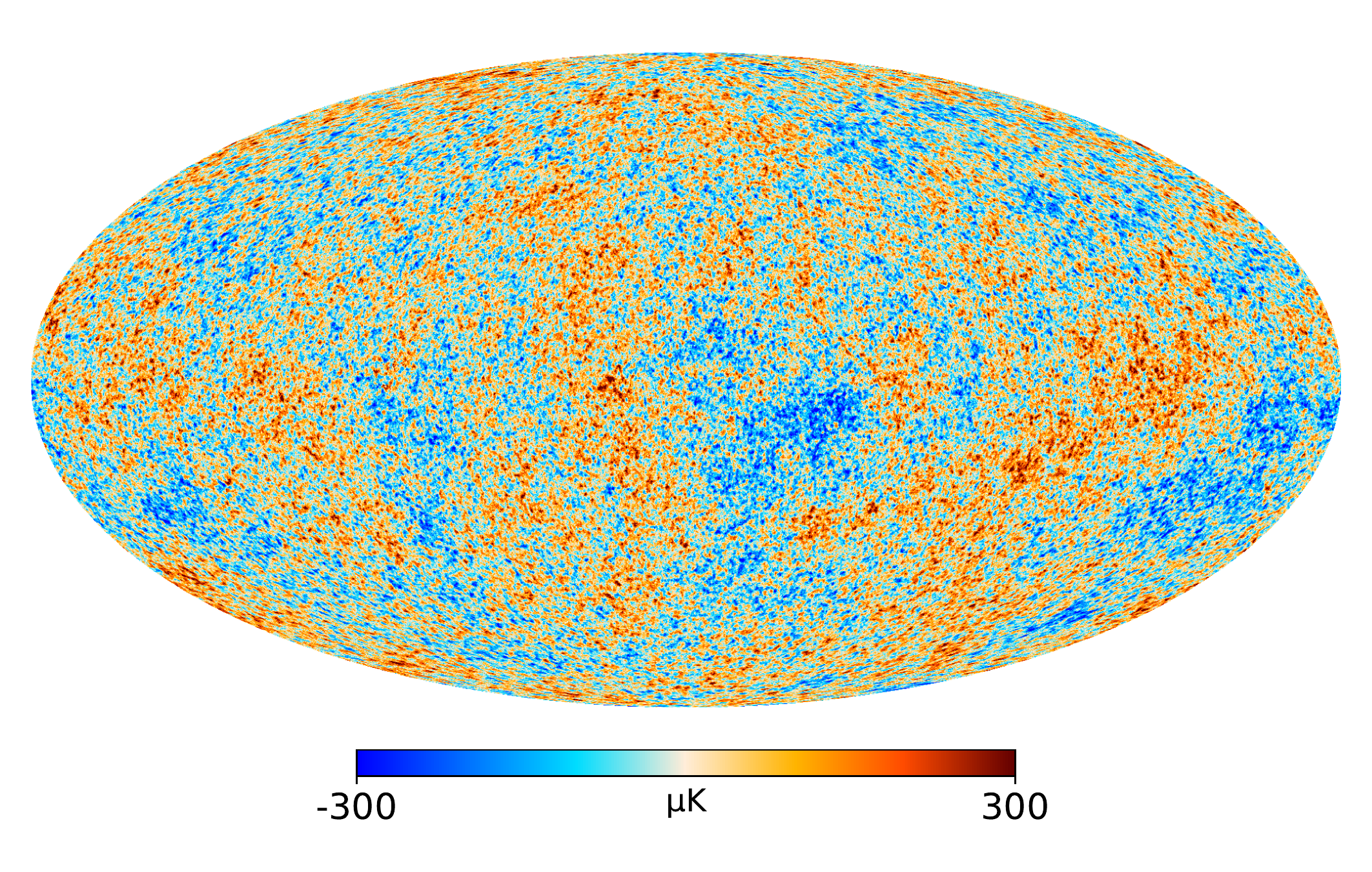}\vspace{-0.5cm}
\caption{CMB Temperature field as measured by the Planck satellite. This is a snapshot of the Early Universe, allowing us to study the components and evolution of the Universe. This spherical field is projected here with a Mollweide projection in galactic coordinates. }
\label{f:planck}
\end{figure}

\section{MATHEMATICAL FOUNDATIONS}

\subsection{Spectral Representation for Isotropic Spherical Random Fields}

Formally, a spherical random field is a collection of real random variables
indexed by points on the sphere, \textit{i.e.}, $\left\{ T:\Omega \times \mathbb{S}%
^{2}\rightarrow \mathbb{R}\right\} ,$ for some suitable underlying
probability space $\left\{ \Omega ,\Im ,\mathbb{P}\right\} ;$ without loss
of generality we take these random variables to be zero mean $\mathbb{E}\left[
T(x)\right] =0,$ we assume they have finite variance $\mathbb{E}\left[
T^{2}(x)\right] <\infty ,$ and we assume isotropy, meaning that the field is
invariant in law to rotations:%
\begin{equation*}
T(x)\overset{Law}{=}T(gx)\text{ for all }g\in SO(3).
\end{equation*}%
Isotropy can be viewed as broadly analogous to the (strong) stationarity
that allows to develop the classical spectral representation approach when
dealing with Time Series, see \textit{e.g.} \cite{BD}. Indeed, under these
conditions, a Spectral Representation Theorem holds on the sphere; more
precisely, we have \citep[see the Appendix or][]{marpecbook}%
\begin{equation}
T(x,\omega )=\sum_{\ell =0}^{\infty }\sum_{m=-\ell }^{\ell }a_{\ell
m}(\omega )Y_{\ell m}(x)\text{ ,}  \label{SRT2}
\end{equation}%
where the identity holds both in $L^{2}(\Omega )$ and in $L^{2}(\Omega
\times \mathbb{S}^{2})$; here we have written explicitly the random
variables as functions of $\omega \in \Omega $ to highlight the decoupling
between random components depending on $\Omega $ and deterministic
components depending on $\mathbb{S}^{2}$.

The deterministic functions $\left\{ Y_{\ell
m}\right\}$ are the \emph{spherical harmonics} and, for fixed $\ell$, they form an orthonormal basis of the eigenspace of the Spherical Laplacian operator $\Delta _{\mathbb{S}^{2}}$, with
eigenvalues corresponding to $-\lambda _{\ell }=-\ell (\ell +1)$. They indeed satisfy%
\begin{equation}
\Delta _{\mathbb{S}^{2}}Y_{\ell m}=-\lambda _{\ell }Y_{\ell m}\text{ , }\ell
=0,1,2,...,\text{ }m=-\ell ,...,\ell \text{ ;}  \label{eig2}
\end{equation}%
again, more details and definitions are given in the Appendix. We therefore have
a decomposition into a set of eigenfunctions, along different frequencies
basically corresponding to the square root of the corresponding eigenvalues
(known as \emph{multipoles }$\ell $ in the spherical case). The most
important message to remember, also for its statistical consequences, is
that in the case of the sphere there are $2\ell +1$ independent components
corresponding to each ``frequency'' (multipole) $\ell$: this fact will play a
crucial role in the asymptotic results to follow.

By expressing the field $T(x)$ in the spherical harmonics base according to Equation \ref{SRT2}, we obtain the \emph{random spherical harmonics coefficients }$\left\{ a_{\ell m}\right\}$. They form an array of zero-mean, complex valued random variables with covariance given by%
\begin{equation*}
\mathbb{E}\left[ a_{\ell m}\overline{a_{\ell ^{\prime }m^{\prime }}}\right]
=\delta _{\ell }^{\ell ^{\prime }}\delta _{m}^{m^{\prime }}C_{\ell }\text{ ;}
\end{equation*}%
in other words, they are uncorrelated whenever $\ell $ or $m$ differ from $%
\ell ^{\prime },m^{\prime }$ and they have variance given by the
non-negative sequence $\left\{ C_{\ell }\right\}$, called the \emph{angular power
spectrum} of the random field $T(x)$. The covariance function of the field can be
expressed as (Schoenberg's Theorem):
\begin{equation*}
Cov(T(x_{1}),T(x_{2}))=\mathbb{E}\left[ T(x_{1}),T(x_{2})\right] =\sum_{\ell
=0}^{\infty }\frac{2\ell +1}{4\pi }C_{\ell }P_{\ell }(\left\langle
x_{1},x_{2}\right\rangle )\text{ ,}
\end{equation*}%
where $P_{\ell }(.)$ is the sequence of Legendre polynomials, to be
introduced in the Appendix and $\langle x,y \rangle$ is the standard scalar product.

It should also be noted that the spherical harmonic coefficients can in
principle be recovered from the observations of the field (up to some
statistical difficulties that we shall discuss below) by means of the
inverse spherical harmonic transform:%
\begin{equation}
a_{\ell m}=\int_{\mathbb{S}^{2}}T(x)\overline{Y_{\ell m}(x)}dx\text{ .}
\label{SHTransform}
\end{equation}

From a cosmological perspective, the angular power spectrum of the CMB, $C_\ell$, is a main source of information. It is sensitive to the exact values of the cosmological parameters, such as the amount of baryonic matter, Dark Matter, or Dark Energy in the Universe \citep{planckPS, planckCosmo}. Modern CMB observations such as Planck are able to measure the angular power spectrum with unprecedented precision, leading to subpercentage uncertainties in most cosmological parameters; because of this, it is commonly said that the field is now in an era of \textit{precision Cosmology}.

\subsection{On the Meaning of Asymptotics in Cosmology}

A deep foundational question must be addressed when dealing with
cosmological data. By definition, Cosmology is a science based on a single
observation---our Universe; then, how is it possible to apply asymptotic
statistics tools?

The asymptotic theory which is used in this framework is meant in the
high-frequency sense, rather than the large domain one; it is strictly
related to so-called fixed-domain asymptotics, the common framework under
which geophysical data are usually handled \cite[see, \textit{e.g.},][and the references therein]{loh,loh2}. It is assumed that inverse Fourier transforms
like Equation \ref{SHTransform} are evaluated for $\ell =1,2,...,L_{\max },$ where $%
L_{\max }$ grows from one experiment to the other. This setting fits exactly
the reality of data collection in Cosmology: a pioneering experiment like
COBE had a resolution to reach an $L_{\max }$ of the order of 20--30, a
number which was raised to 600--800 for WMAP and subsequently to about 2000
by Planck; in the next generation of experiments, these values could
grow further by a factor of at least 2 or 3. More precisely, observations
are collected on a grid of sampling points that depend on the resolution of
the experiment: it arrives to a cardinality of about $12\times 1024^{2}$
points for Planck. By means of these observations, integrals as Equation \ref%
{SHTransform} are approximated as Riemann sums; the order of the multipoles
for which this approximation can work depends on the resolution of the grid.

In practice, however, the observations are not available with the same level
of accuracy on the full sphere: there are some regions of the sky where foreground
contaminants such as the Milky Way cannot be removed efficiently. It is then
convenient to introduce some form of spherical wavelet transform, as we
shall do in the remainder of this section.

In particular, we will consider needlets: a form of spherical wavelets which
was introduced by \citet{npw1,npw2} in the mathematical literature \citep[see also][]{gm1,gm2,gm3} and then to the Cosmological and Statistical
communities by \cite{bkmpAoS,BKMP,mpbb08}, see also \cite{scodellermexican} and the references therein. The basic idea behind needlets can be summarized as follows. Consider first the operator which goes from the random field $T(x)$ to its Fourier components $T_{\ell
}(x)$:
\begin{eqnarray*}
T_{\ell }(x) &=&\sum_{m=-\ell }^{\ell }a_{\ell m}Y_{\ell m}(x)=\sum_{m=-\ell
}^{\ell }\int_{\mathbb{S}^{2}}T(z)\overline{Y}_{\ell m}(z)dzY_{\ell m}(x) \\
&=&\int_{\mathbb{S}^{2}}T(z)\sum_{m=-\ell }^{\ell }\overline{Y}_{\ell
m}(z)Y_{\ell m}(x)dz=\int_{\mathbb{S}^{2}}T(z)\frac{2\ell +1}{4\pi }P_{\ell
}(\left\langle z,x\right\rangle )dz\text{ ,}
\end{eqnarray*}%
where the last equality comes from Equation \ref{addition} in the Appendix. In other words, the Fourier
components can be realized as the projections of the fields on the ``frequency'' (multipole) component $\ell $ by means of the kernel operator $\frac{2\ell +1}{4\pi }P_{\ell }(\left\langle x,.\right\rangle )$. These multipole components of the field $T_\ell$ are perfectly localized in harmonic domain (as they are projected onto a single multipole) at the expense of losing all spatial localization. This can be a problem when different parts of the sphere are observed with different levels of noise or with incomplete observations.

The idea of needlets is to partially give up the perfect
localization in the harmonic domain in order to obtain better localization properties in the real domain. In particular, let us now consider a smooth (possibly $C^{\infty }$)
non-negative function $b(.)$, supported in the compact domain $(\frac{1}{B}%
,B),$ with $B>1$; we additionally impose that its square satisfies the partition of
unity property, namely
\begin{equation*}
\sum_{j =0}^{\infty }b^{2}\left(\frac{\ell }{B^{j}}\right)\equiv 1\text{ , for all }%
\ell=1,2,....
\end{equation*}%
Note that the elements in the sum are different from zero only for $\ell \in
(B^{j-1},B^{j+1}).$ Now consider also a grid of cubature points $\left\{ \xi
_{jk}\in \mathbb{S}^{2}\right\} _{j=1,2,...;k=1,2,...N_{j}}$, with growing
cardinality $N_{j}$ of order $B^{2j},$ nearly equispaced so that $d_{\mathbb{%
S}^{2}}(\xi _{jk},\xi _{jk^{\prime }})\simeq cB^{-j},$ where we have
introduced the spherical geodesic distance%
\begin{equation*}
d_{\mathbb{S}^{2}}(x,y):=\arccos \left( \left\langle x,y\right\rangle
\right) \text{ .}
\end{equation*}%
The needlet projection coefficients are defined by%
\begin{equation*}
\beta _{jk}:=\int_{\mathbb{S}^{2}}T(x)\psi _{jk}(x)dx\text{ , }\psi
_{jk}(x):=\sum_{\ell }b\left(\frac{\ell }{B^{j}}\right)\frac{2\ell +1}{4\pi }P_{\ell
}(\left\langle x,\xi _{jk}\right\rangle )\sqrt{\lambda _{jk}}\text{ ,}
\end{equation*}%
where $\left\{ \lambda _{jk}\right\} $ are cubature weights, which can be
chosen in such a way as to ensure the identity%
\begin{equation*}
\sum_{k=1}^{N_{j}}Y_{\ell m}(\xi _{jk})\overline{Y_{\ell ^{\prime
}m^{\prime }}}(\xi _{jk})\lambda _{jk}=\delta _{\ell }^{\ell ^{\prime
}}\delta _{m}^{m^{\prime }}\text{ , for all }\ell ,\ell ^{\prime }\leq
\lbrack B^{j+1}]\text{ .}
\end{equation*}%
It is then easy to see that the following reconstruction formula holds:%
\begin{equation*}
T(x)=\sum_{jk}\beta _{jk}\psi _{jk}(x)=\sum_{j}\widetilde{T}_{j}(x)\text{ ,}
\end{equation*}%
where the needlet components are defined by%
\begin{equation}
\widetilde{T}_{j}(x)=\sum_{\ell =0}^{\infty }b^{2}\left(\frac{\ell }{B^{j}}\right)T_{\ell }(x)\text{ .}  \label{needcomponents}
\end{equation}
Alternatively, we could introduce the \emph{needlet projection kernel}%
\begin{equation*}
\Psi _{j}(x,y)=\sum_{\ell =0}^{\infty }b^{2}\left(\frac{\ell }{B^{j}}\right)\frac{2\ell
+1}{4\pi }P_{\ell }(\left\langle x,y\right\rangle )
\end{equation*}%
which acts on the field $T(.)$ in such a way that%
\begin{eqnarray*}
\Psi _{j} :T(x)\rightarrow &&\int_{\mathbb{S}^{2}}T(z)\Psi _{j}(z,x)dz =\\
&=&\sum_{\ell =0}^{\infty }b^{2}\left(\frac{\ell }{B^{j}}\right)\int_{\mathbb{S}%
^{2}}T(z)\frac{2\ell +1}{4\pi }P_{\ell }(\left\langle z,x\right\rangle )dz=%
\widetilde{T}_{j}(x)\text{ .}
\end{eqnarray*}%
Therefore, it is now readily seen that this needlet projection operator projects the
field $T(.)$ on a linear combination of eigenfunctions $T_{\ell },$ for $%
\ell \in (B^{j-1},B^{j+1}).$ More importantly, the needlet kernel projector
enjoys much better localization properties than the simple Legendre
projector; indeed, it has been shown \cite[see][]{npw1,npw2,gm1} that for all $M\in \mathbb{N}$ there exists a positive constant $C_{M}$ s.t.%
\begin{equation}
\left\vert \Psi _{j}(x,y)\right\vert \leq C_{M}\frac{B^{2j}}{(1+B^{j}d_{%
\mathbb{S}^{2}}(x,y))^{M}}\text{ .}  \label{needkern}
\end{equation}
In words, this means that for any fixed angular distance, the kernel decays to zero faster than any polynomial in $j.$ As a consequence, it is
possible to evaluate needlet projections even in the case of sky maps which
are only partially observed, given that, for any ``masked'' (\textit{i.e.}, unobservable)
region $G\subset \mathbb{S}^{2}$ and $x\in \mathbb{S}^{2}\backslash G$ \ we
have%
\begin{eqnarray*}
\int_{\mathbb{S}^{2}\backslash G}T(z)\Psi _{j}(z,x)dz &=&\int_{\mathbb{S}%
^{2}}T(z)\Psi _{j}(z,x)dz-\int_{\mathbb{G}}T(z)\Psi _{j}(z,x)dz \\
&=&\widetilde{T}_{j}(x)+R_{j}(x)\text{ ,}
\end{eqnarray*}%
where%
\begin{eqnarray*}
\mathbb{E}\left[ \left\vert R_{j}(x)\right\vert \right] &\leq &\int_{G}%
\mathbb{E}\left[ \left\vert T(z)\right\vert \right] \left\vert \Psi
_{j}(z,x)\right\vert dz \\
&\leq &Const\times \int_{G}\left\vert \Psi _{j}(z,x)\right\vert dz \\
&\leq &Const\times C_{M}\int_{G}\frac{B^{2j}}{(1+B^{j}d_{\mathbb{S}%
^{2}}(x,z))^{M}}dz \\
&\leq &C_{M}^{\prime }\times 4\pi \times B^{j(2-M)}\times d_{\mathbb{S}%
^{2}}^{-M}(x,G)\rightarrow 0\text{ , as }j\rightarrow \infty \text{ ,}
\end{eqnarray*}%
where $d_{\mathbb{S}^{2}}(x,G)$ is defined as the infimum of the distances
between $x$ and the points of $G$. The decay to zero is itself
super-exponential, and a very broad numerical and empirical evidence has
shown that needlet components of spherical random fields are minimally
affected by unobserved regions, in the high-frequency sense, and for
reasonable amounts of masked information.

Needlet fields enjoy another very important property for statistical
analysis. In particular, it has been established that \cite[see][]{BKMP}%
\begin{equation*}
\left\vert Corr(\widetilde{T}_{j}(x),\widetilde{T}_{j}(y))\right\vert \leq
Const\times \frac{1}{(1+B^{j}d_{\mathbb{S}^{2}}(x,y))^{M}}\text{ ;}
\end{equation*}%
in other words, the correlation between the needlet fields evaluated on any
two points decays faster than any polynomial. At first sight, one may think
this is a direct consequence of the localization properties of the needlet
kernel projector (see Equation \ref{needkern}), but this is not the case:
uncorrelation is not in general a consequence of kernel localization. To
understand this point, consider the extreme case of a delta-like projector
such that $\delta _{x}:T\rightarrow T(x)$; this is obviously an example of
perfect localization in real space, but nevertheless the correlation
between any two projected components $T(x)$ and $T(y)$ does not decay to
zero in any meaningful sense. It is the combination of localization in real
and harmonic space, a defining feature of needlets, that makes fast
uncorrelation possible.

In the next Section, we show how these uncorrelation properties make it possible to study principled statistical inference with an asymptotic justification.

\section{CMB MAP RECONSTRUCTION AND COMPONENT SEPARATION}
The first statistical issue we shall consider is the so-called CMB map reconstruction (sometimes also called component separation or foreground removal). This is the cosmological instance of the image
reconstruction issues that are common in many fields, and because of this, the
techniques introduced in this chapter are likely to be applicable to
many different areas.

The problem arises from a very natural question: when observing the
celestial sky, how can you distinguish CMB radiation from the many other
galactic and extragalactic sources that lie between us and the last
scattering surface? The key remark is that CMB observations
are collected on many different electromagnetic frequencies, where they
follow a Planckian black-body emission governed by a single (temperature)
parameter \cite[see][]{Durrer,planckMM,axelsson}. Therefore, we can decompose the observation at frequency $\nu_k$ and point $x$, $T(x,\nu_k)$, after suitable transformations, in the following way:%
\begin{equation*}
T(x;\nu _{k})=T(x)+F(x;\nu _{k})+N(x,\nu _{k})\text{ ,}
\end{equation*}%
where $N(x,\nu _{k})$ denotes instrumental noise and $F(x;\nu _{k})$ are
``foreground residuals'', \textit{i.e.} the collection of emission by galactic dust,
astrophysical sources, and other mechanisms both from within the Galaxy and from outside it. The crucial identifying assumption is that the CMB ``signal'' $T(x)$ is
constant across the different electromagnetic channels (although by no means
constant over the sky directions $x\in \mathbb{S}^{2}$). Assume as a
starting point that the noise $N(.)$ has zero mean, constant variance and is
uncorrelated over different electromagnetic frequencies; for the foreground,
assume again that it has finite variance, and that its variance covariance
matrix over different channels is given by%
\begin{equation*}
\mathbb{E}\left[ F(x;.)F(x;.)^{T}\right] =\Omega \text{ .}
\end{equation*}%
The best linear unbiased estimates of $T(x),$ viewed as a fixed parameter at
a given $x$, is then simply given by the generalized least squares solution.%
\begin{equation*}
\widehat{T}_{ILC}(x)=\left\{ P^{T}\left( \Omega +\sigma _{N}^{2}I_{K}\right)
^{-1}P\right\} ^{-1}P^{T}\left( \Omega +\sigma _{N}^{2}I_{K}\right)
^{-1}T(x,.)\text{ ,}
\end{equation*}%
where $P=(1,1,...,1)^{T}$ is the $K$-dimensional vector of ones, and by $%
T(x,.)$ we mean the $K$-dimensional column vector with the observations
across the different frequencies. This map-making algorithm is known as ILC
(Internal Linear Combination) in the cosmological literature. To be
implemented, it requires an estimate of the covariance matrix $%
\Omega$; here is where the needlet approach can come into play.

The covariance matrix $\Omega $ represents the dependence structure and the
overall magnitude of different foregrounds over different frequencies. This
covariance varies wildly over different regions and over different scales:
some foregrounds are dominant over large scales (\textit{e.g.}, galactic dust) whereas other
dominates on smaller scales (localized sources). To take into account the
fact that variance can be different across different scales, the idea is to
introduce the $NILC$ (Needlet Internal Linear Combination) map--making
algorithms, defined by \cite[see][and references therein]{Delabrouille}
\begin{equation*}
\widehat{T}_{NILC;j}(x)=\left\{ P^{T}\left( \widehat{\Omega }_{j}+\sigma
_{N}^{2}I_{K}\right) ^{-1}P\right\} ^{-1}P^{T}\left( \widehat{\Omega }%
_{j}+\sigma _{N}^{2}I_{K}\right) ^{-1}\widetilde{T}_{j}(x,.)\text{ ,}
\end{equation*}%
\begin{equation*}
\widehat{T}_{NILC}(x)=\sum_{j}\widehat{T}_{NILC;j}(x)\text{ ,}
\end{equation*}%
where the covariance matrix $\Omega _{j}$ is estimated in a first step as%
\begin{equation*}
\widehat{\Omega }_{j}=\frac{1}{4\pi }\int_{\mathbb{S}^{2}}\left( T(x;\nu
_{k})-\overline{T}(x)\right) \left( T(x;\nu _{k})-\overline{T}(x)\right)
^{T}dx\text{ , }\quad\overline{T}(x)=\frac{1}{K}\sum_{k}T(x;\nu _{k})\text{ .}
\end{equation*}%

As a second step (see, \textit{e.g.}, \cite{Car2} and the references therein), one may
take into account the spatial variability of the matrices $\Omega =\Omega
(x)$, another task that can be addressed by needlets because of their
spatial localization. The idea is then to partition the celestial sky into
subregions $A_{r}$, with $r=1,...,R$, such that $A_{r_{1}}\cap A_{r_{2}}=\varnothing$ and $%
\cup _{r=1}^{R}A_{r}=\mathbb{S}^{2}$. We can write now%
\begin{equation*}
\widehat{\Omega }_{j,r}=\frac{1}{4\pi }\int_{A_{r}}\left( T(x;\nu _{k})-%
\overline{T}(x)\right) \left( T(x;\nu _{k})-\overline{T}(x)\right) ^{T}dx
\end{equation*}%
\begin{equation*}
\widehat{T}_{NILC;j,r}(x)=\left\{ P^{T}\left( \widehat{\Omega }_{j,r}+\sigma
_{N}^{2}I_{K}\right) ^{-1}P\right\} ^{-1}P^{T}\left( \widehat{\Omega }%
_{j,r}+\sigma _{N}^{2}I_{K}\right) ^{-1}\widetilde{T}_{j}(x,.)\text{ , }\quad x\in
A_{r}\text{ ,}
\end{equation*}%
\begin{equation*}
\widehat{T}_{NILC;r}(x)=\sum_{j}\widehat{T}_{NILC;j,r}(x)\text{ , }\quad x\in A_{r}%
\text{ .}
\end{equation*}%
This method has proven to be very efficient and compares favorably with
other existing techniques, see again \cite{Car2} for further discussion.

For the rest of this paper, we shall assume to be dealing with maps that
have been made according to one of these procedures: in the next two
sections, we discuss how to test for residuals point sources and for
deviations from the assumptions of isotropy and Gaussianity.

\section{POINT SOURCE DETECTION AND SEARCH FOR GALAXY CLUSTERS}
\label{s:ps}
After a CMB map has been built from observations, as discussed in the previous Section, a
natural question to ask is whether all the foreground contaminants have been
properly removed. In this Section, we focus on so-called point sources:
these are mainly galaxies or clusters of galaxies which appear as point-like
local maxima in the maps \citep{planckCS}.

The proper approach to handling such issues is clearly a form of multiple
testing. Indeed, for an experiment like Planck, there are several thousand
local maxima that could be identified as potential candidates for point
sources: at each of these locations one may wish to run a significance test,
but controlling the size of the test is then a daunting task. Very recently,
this topic has been addressed in \cite{CarronDuque,ChengCammarota} by
means of a variation in the so-called STEM (Smoothing and Testing Multiple
Hypotheses) algorithm, see also \cite{Schwartzman:2011,chengschwartzman1,chengschwartzman2,chengschwartzman2015m,cmw2014}.

The idea of the procedure is rather natural and can be explained as follows.
Because we are looking for point sources, large-scale fluctuations can be
considered to produce just noise without information on the signal; in the
first step, it is then natural to filter the map with a needlet transform and
consider only the needlet components $\widetilde{T}_{j}$ (as defined in Equation \ref{needcomponents}), for ``high enough'' $j$, as we will discuss later. This
procedure does increase the signal-to-noise ratio considerably, as
illustrated in \textbf{Figure \ref{f:ps}}, see also \cite{scodeller, scodeller2}.

\begin{figure}
\centering
\begin{minipage}{0.48\textwidth}
    \includegraphics[width=\linewidth]{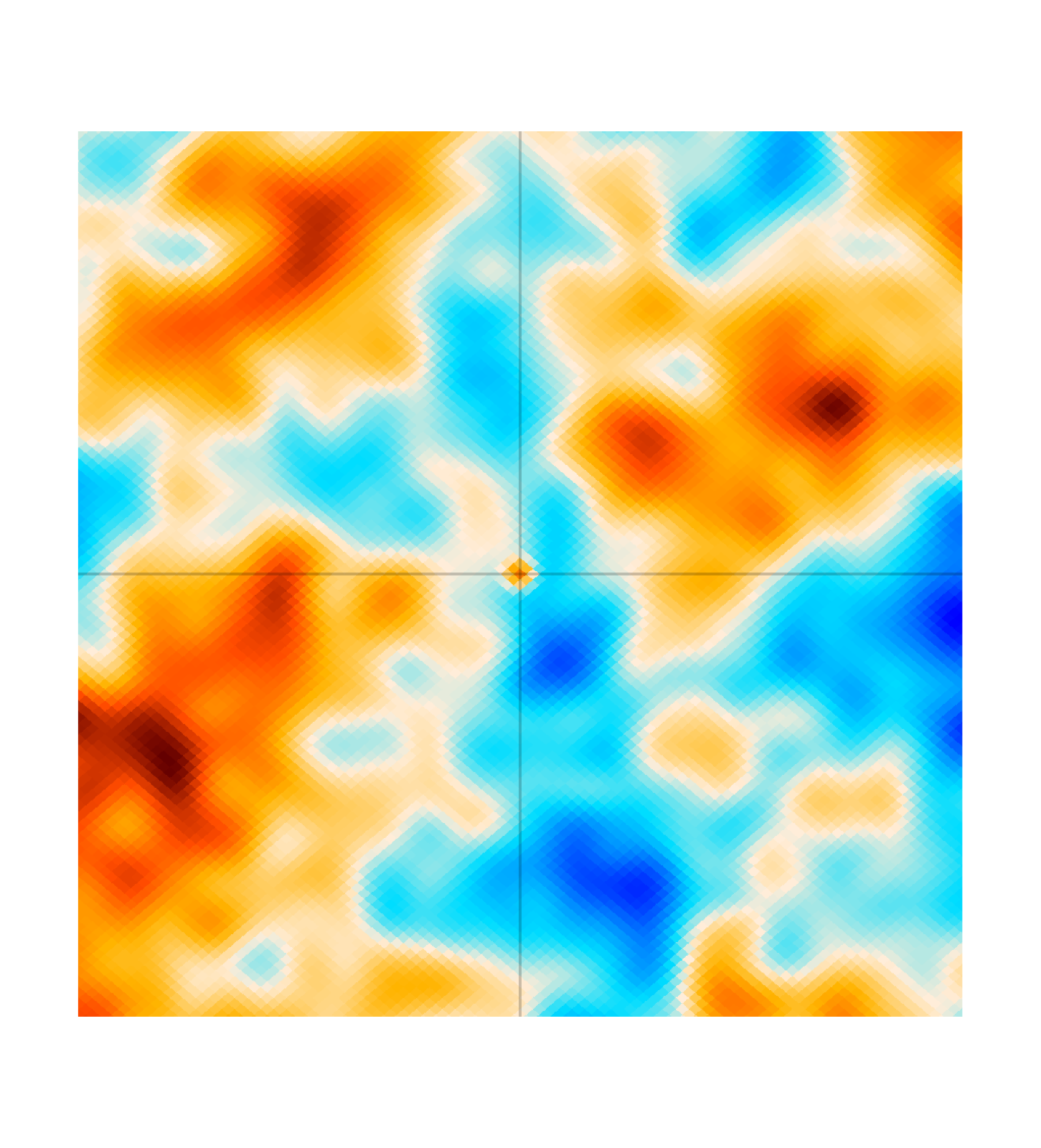}
\end{minipage}
\begin{minipage}{0.48\textwidth}
    \includegraphics[width=1\linewidth]{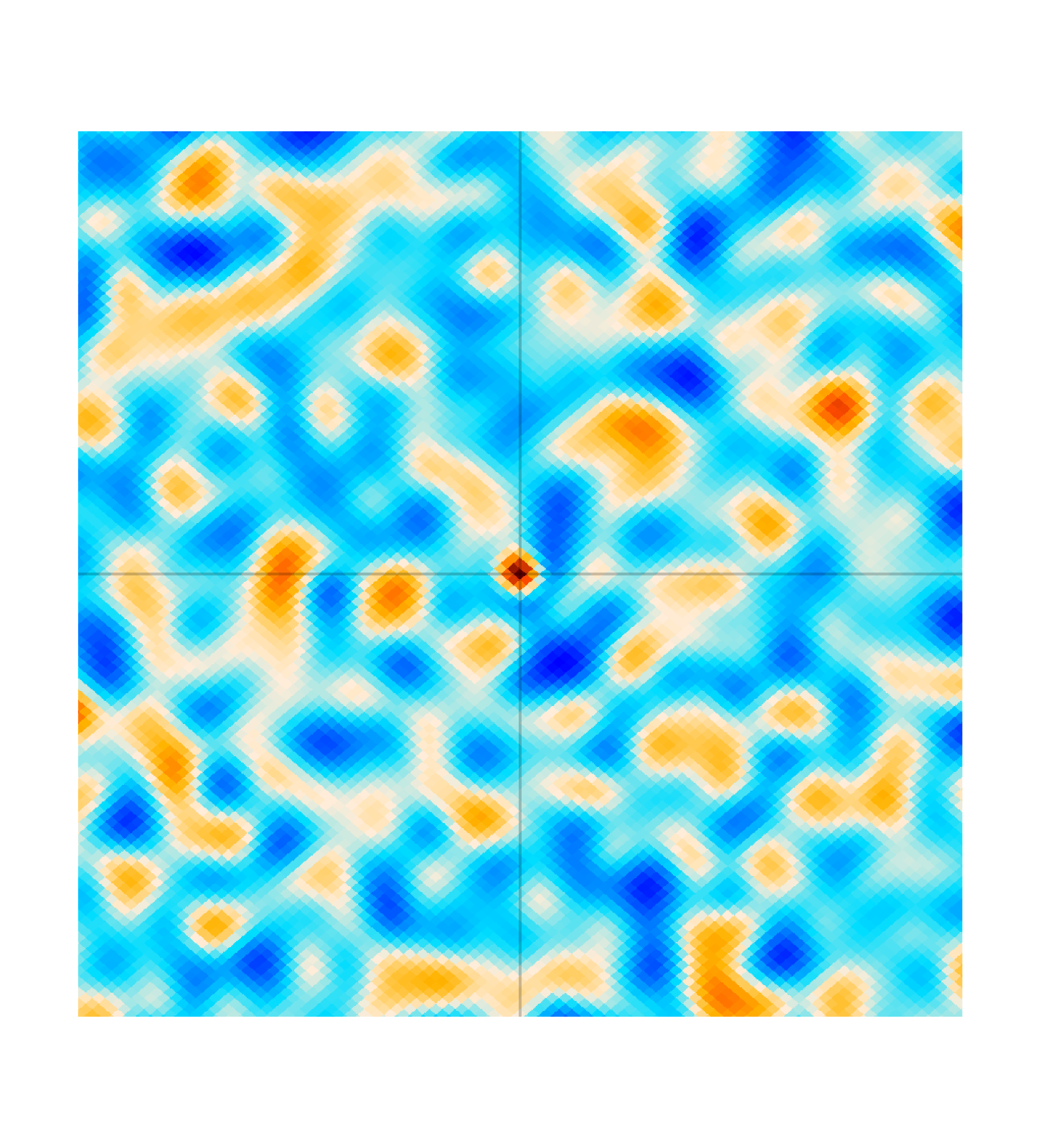}
\end{minipage}
\caption{Effect of a point source on the CMB. On the left, CMB Temperature $T(x)$; on the right, a single needlet component, $\widetilde{T}_{j}$. A point source has been introduced in the center of the map. The signal--to--noise ratio at the peak in this example goes from $0.4$ to $4.4$. Both images are $3^\circ \times3^\circ$ patches of the sky.}
\label{f:ps}
\end{figure}

In a second step, we focus on the candidate sources, which correspond to the
local maxima of the field $\widetilde{T}_{j}$. Our aim is to establish the $p$-value of each of these local maxima. Therefore, we need to study the density of these maxima. It turns out that an exact result can be given \cite[see][]{ChengCammarota}; we shall denote the probability density of maxima by $f_j$.

The expected density of maxima $f_j$ is compared to the observed distribution in a single realization in \textbf{Figure \ref{f:psdist}}, for a needlet-filtered map at $B=1.2$, $j=39$. At these small scales, corresponding to the size of point sources, the theoretical expectaction is remarkably close to the realized distribution, even when considering a single map.

\begin{figure}
\centering
    \includegraphics[width=1\linewidth]{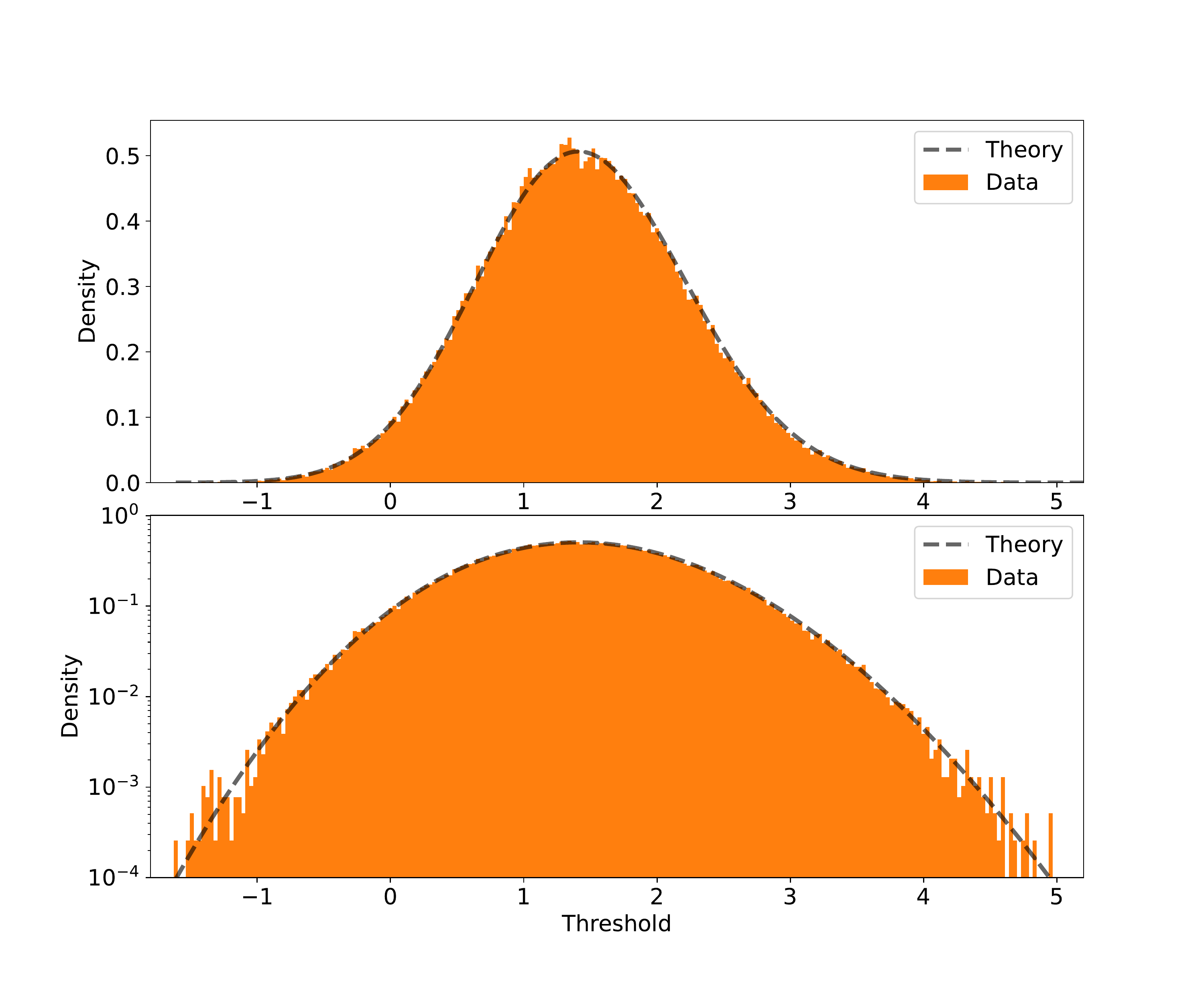}
\caption{Distribution of the maxima of a needlet-filtered map, $\widetilde{T}_{j}$ with a small-scale needlet ($B=1.2$, $j=35$). In dashed black line, the theoretical prediction, $f_j$; in orange bars, the realization on a single map. Top: linear scale. Bottom: logarithmic scale. It can be seen that the maxima density closely follows the prediction, even in a single realization.}
\label{f:psdist}
\end{figure}

We are then in the position to compute the $p$-value of each maximum: it is the probability of the maximum being larger than $\widetilde{T}_{j}(\xi _{k})$, with $k=1,2,...K$ for the $K$ maxima located at $%
\xi _{1},...\xi _{K}\in \mathbb{S}^{2}$. Therefore, the $p$-value of each maximum is
\begin{equation*}
p_{k,j}:=\int_{\widetilde{T}_{j}(\xi _{k})}^\infty f_{j}(u)du\text{ .}
\end{equation*}%
As a third step, we implement a Benjamini-Hochberg procedure for the
identification of point sources \citep{BH}. To this aim, let us first reorder the $p$%
-values and sources so that they are increasing, $p_{(1),j}\leq
p_{(2),j}\leq ...\leq p_{(K),j};$ we then fix a value $0<\alpha <1,$
corresponding to the expected proportion of False Discoveries that we are
willing to tolerate among the reported Point Sources. Namely, we consider $\xi _{k}$ to be the location of a Point Source whenever%
\begin{equation*}
u_{k,j}:=p_{(k),j}-\alpha \frac{k}{K}<0\text{ ,}
\end{equation*}%
that is, when the $(k)$\textsuperscript{th} ordered $p$-value is smaller that $\alpha $ times the expected $p$-value of the $(k)$\textsuperscript{th} maximum under the null hypothesis of a purely Gaussian field.

The technical analysis of the properties of this procedure is based on a
result of some independent interest, namely the high-frequency ergodicity of
the empirical distribution of the maxima. In other words, it turns out to be
possible to show that, as $j\rightarrow \infty ,$ for all fixed $u\in
\mathbb{R}$%
\begin{equation*}
\frac{Card\left\{ \widetilde{T}_{j}(\xi _{k}):\widetilde{T}_{j}(\xi
_{k})>u\right\} }{\int_{u}f_{j}(x)dx}\rightarrow _{p}1\text{ .}
\end{equation*}%
In turn, this result follows from the Kac--Rice representation of maxima and
a high-frequency uncorrelation result on the gradient and Hessian of the
needlet fields, see \cite{ChengCammarota} for more discussion and details.

The previous result forms the mathematical basis for establishing the two
main properties of the Benjamini-Hochberg procedure in this context, namely
\begin{enumerate}
    \item[(a)] False Discovery Rate control: as $j\rightarrow \infty$, we have that%
\begin{equation*}
\mathbb{E}\left[ \frac{Card\left\{ \widetilde{T}_{j}(\xi _{k}):u_{k,j}<0%
\text{ and no point source at }\xi _{k}\right\} }{Card\left\{ \widetilde{T}%
_{j}(\xi _{k}):u_{k,j}<0\right\} }\right] <\alpha
\end{equation*}
    \item[(b)] Power control: as $j\rightarrow \infty ,$ the proportion of point sources
that are detected converges to unity.
\end{enumerate}



Actually, the result established in \cite{ChengCammarota} is stronger than
(b), because the number of sources is allowed to grow with $j,$ in an effort
to mimic in this context a high-dimensional framework where the complexity
of the signal (point sources) grows with the amount of observations. The
approach described here was successfully applied to Planck CMB temperature
maps and led to the possible discovery of a previously undetected source, see \cite%
{CarronDuque} for more details.

A final remark on this topic: in this section, the procedure we have introduced
is based on the framework where the signal is made up by the sources to be
detected, whereas noise is the given by a random Gaussian field, namely the
CMB radiation which we consider as signal for most of the paper. This is
sometimes summarized by the ironical motto \textit{``your noise is my signal''} in some
papers in this area. In the next section, we will further probe goodness-of-fit tests for the basic assumptions on the random fields at hand, namely Gaussianity and isotropy.

\section{TESTING FOR GAUSSIANITY AND ISOTROPY}

A crucial assumption to be investigated on spherical random fields is
whether they are actually Gaussian and isotropic. A very important tool
for this task is the so-called \textit{Lipschitz--Killing Curvatures}; in the cosmological literature, the term ``Minkowski functionals'' is more often used, but the two
definitions are equivalent up to some scaling constants and reordering of indexes.

The proper definition of Lipschitz-Killing Curvatures requires some
geometric background. Let us first consider some convex set $A\subset
\mathbb{R}^{d}$, and let us define the $\emph{Tube}$ of radius $r$ around $A$
as the set of points at distance smaller or equal than $r$ from $A:$%
\begin{equation*}
Tube(A,r):=\left\{ x\in \mathbb{R}^{d}:d(x,A)\leq r\right\} \text{ .}
\end{equation*}%
The \emph{Tube Formula} \citep[see][]{RFG,AT2} proves that the volume
of the Tube admits a finite-order Taylor expansion into powers of $r$:
\begin{equation*}
Meas\left\{Tube(A,r)\right\} =\sum_{i=0}^{d}\,\omega _{d-i}\:\mathcal{L}%
_{i}(A)\:r^{i}\text{ , where }\omega _{i}=\frac{\pi ^{i/2}}{\Gamma (\frac{i}{2}%
+1)}\text{ ,}
\end{equation*}%
$\Gamma (\alpha )=\int_{0}^{\infty }t^{\alpha -1}\exp (-t)dt$ denoting the
Gamma function and $\omega _{i}$ represents the volume of the $i$%
--dimensional unit ball ($\omega _{0}=1,$ $\omega _{1}=2$, $\omega _{2}=\pi$,
$\omega _{3}=\frac{4}{3}\pi$, $\dots$). The coefficients $\mathcal{L}_{i}(A)$
are the so-called Lipschitz-Killing Curvatures of $A$; for instance, in
two-dimensional space there are three of them, equal to the area, half the boundary length, and
the Euler--Poincar\`{e} Characteristic of the set $A$. The latter is a topological invariant
that plays a crucial role in Mathematics: in the two-dimensional case, it is
equivalent to the number of connected components of $A$ minus its ``holes''.
It can be shown that any sufficiently regular functional of $A$ can be
written in terms of the Lipschitz-Killing Curvatures alone; in this sense,
they can be viewed as some sort of sufficient statistic for the information
encoded into $A$, see again \cite{RFG,AT2}.

The Lipschitz-Killing Curvatures (or equivalently the Minkowski Functionals,
which are the same quantities up to some constants and relabelling of
indexes) are one of the most popular statistical tools for Cosmological data
analysis. One reason for their popularity is that, quite surprisingly, it is
possible to give simple analytic forms for their expected values, under
the null assumptions of Gaussianity and isotropy: they are hence very
natural tools for goodness-of-fit tests. More precisely, let us define the
derivative of the covariance function at the origin as
\begin{equation*}
\mu ^{2}:=\left. \frac{\partial }{\partial y}\Gamma (\left\langle
x,y\right\rangle )\right\vert _{y=x}=\sum_{\ell }\frac{2\ell +1}{4\pi }\frac{%
\lambda _{\ell }}{2}C_{\ell }\text{ .}
\end{equation*}%
Let us also introduce the excursion sets $A_{u},$ which are simply those
subsets of the sphere where the field is above some given value $u:$%
\begin{equation*}
A_{u}(f):=\left\{ x\in \mathbb{S}^{2}:f(x)\geq u\right\} \text{ .}
\end{equation*}%
The idea is to compute the Lipschitz-Killing Curvatures on the (random) excursion
sets, and then compare their observed values on real data with their
expectation under Gaussianity and isotropy. It could be imagined that
computing the latter might be a daunting task; on the contrary, it turns out
that a completely explicit expression holds for random fields defined on
general manifolds. In particular, the following Gaussian Kinematic Formula
holds :%
\begin{equation*}
\mathbb{E}\left[ \mathcal{L}_{i}(A_{u}(f))\right] =\sum_{k=0}^{2-i}\flag{k+i}{k}\,
\mathcal{L}_{k+i}(\mathbb{S}^{2}) \, \rho _{k}(u) \, \mu ^{k/2}\text{ },
\end{equation*}%
where we have introduced the flag coefficients%
\begin{equation*}
\flag{d}{k}=\binom{d}{k}\frac{\omega _{d}}{\omega _{k}\,\omega _{d-k}}
\end{equation*}%
and the functions%
\begin{eqnarray*}
\rho _{0}(u)&=& 1-\Phi (u) \text{ , }\\
\rho _{k}(u)&=& \frac{1}{(2\pi )^{\frac{k+1}{2}}}H_{k-1}(u)\exp (-%
\frac{u^{2}}{2})\text{ , for }k\geq 1\text{ ,}
\end{eqnarray*}%
where $\Phi (.)$ denotes the standard Gaussian cumulative distribution whereas
$H_{k}(u)$ stands for the Hermite polynomials, given by%
\begin{equation*}
H_{k}(u)=(-1)^{k}\exp (\frac{u^{2}}{2})\frac{d^{k}}{du^{k}}\exp (-\frac{u^{2}%
}{2})=(-1)^{k}\frac{1}{\phi (u)}\frac{d^{k}}{du^{k}}\phi (u)\text{ ;}
\end{equation*}%
here, $\phi (u)$ denotes as usual the standard Gaussian probability distribution. The first Hermite polynomials are $H_{0}(u)=1$, $H_{1}(u)=u$, $H_{2}(u)=u^{2}-1$.

In particular, we obtain the following values for the case of the excursion
area, the boundary length and the Euler-Poincar\'{e} characteristic:%
\begin{eqnarray*}
\mathbb{E}\left[ \mathcal{L}_{2}(A_{u}(f))\right] &=&4\pi (1-\Phi (u))\text{ ,%
} \\
\mathbb{E}\left[ \mathcal{L}_{1}(A_{u}(f))\right] &=&\pi \exp \left(-\frac{u^{2}}{%
2}\right) \, \mu ^{1/2}\text{ ,} \\
\mathbb{E}\left[ \mathcal{L}_{0}(A_{u}(f))\right] &=&2\left[ u \,\phi (u) \,\mu + 1-\Phi(u) \right]
\end{eqnarray*}%
It turns out to be especially convenient to compute Lipschitz-Killing
Curvatures on the needlet components of the random fields, $\widetilde{T}_{j}$. Indeed, in these
circumstances, the expected values are simply obtained by replacing the
needlet covariance derivative for the parameter $\mu$, which is simply
given by:%
\begin{equation*}
\mu _{j}^{2}:=\sum_{\ell }\frac{2\ell +1}{4\pi }\frac{\lambda _{\ell }}{2}\:%
b^{2}\!\left(\frac{\ell }{B^{j}}\right)C_{\ell }\text{ .}
\end{equation*}%
The main advantage is that it can be shown that the variances of
Lipschitz-Killing Curvatures decay to zero and a Central Limit Theorem holds
in the high frequency limit, see, \textit{e.g.}, \cite{CM2015}, \cite{SheTo} and the
references therein. We have that%
\begin{equation*}
\frac{\mathcal{L}_{i}(A_{u}(\widetilde{T}_{j}))-\mathbb{E}\left[ \mathcal{L}_{i}(A_{u}(\widetilde{T}_{j}))%
\right] }{\sqrt{Var\left[ \mathcal{L}_{i}(A_{u}(\widetilde{T}_{j}))\right] }}\rightarrow
_{d}N(0,1)\text{ , as }B^{j}\rightarrow \infty \text{ , }i=0,1,2\text{ .}
\end{equation*}%
Moreover, we have
\begin{equation*}
Var\left[ \frac{\mathcal{L}_{i}(A_{u}(\widetilde{T}_{j}))}{\mathbb{E}\left[ \mathcal{L}%
_{i}(A_{u}(\widetilde{T}_{j}))\right] }\right] =O\left(\frac{1}{B^{2j}}\right)\text{ , as }j\rightarrow
\infty \text{ .}
\end{equation*}%

\begin{figure}
\centering
\begin{minipage}{0.48\textwidth}
    \includegraphics[width=\linewidth]{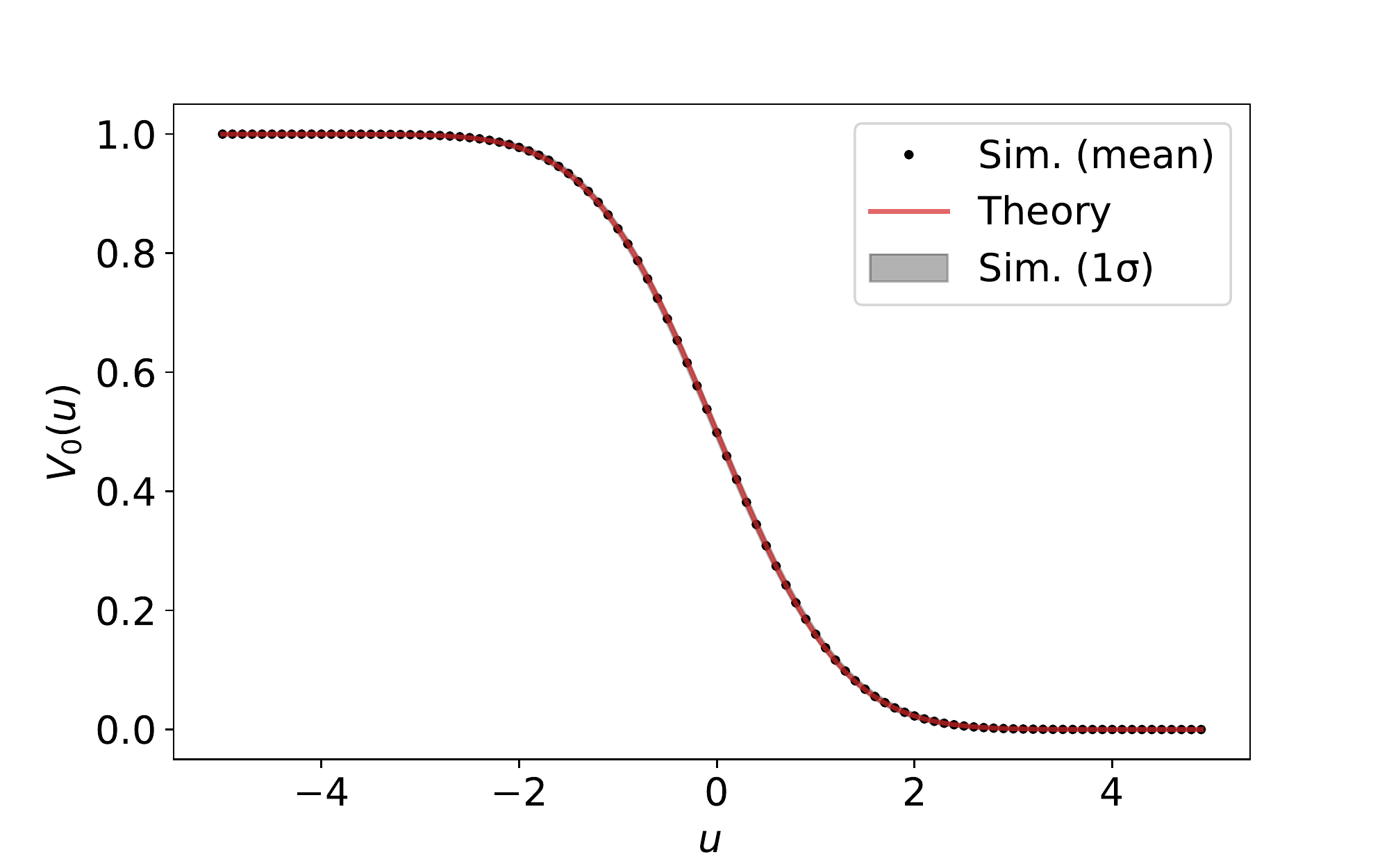}
\end{minipage}
\begin{minipage}{0.48\textwidth}
    \includegraphics[width=1\linewidth]{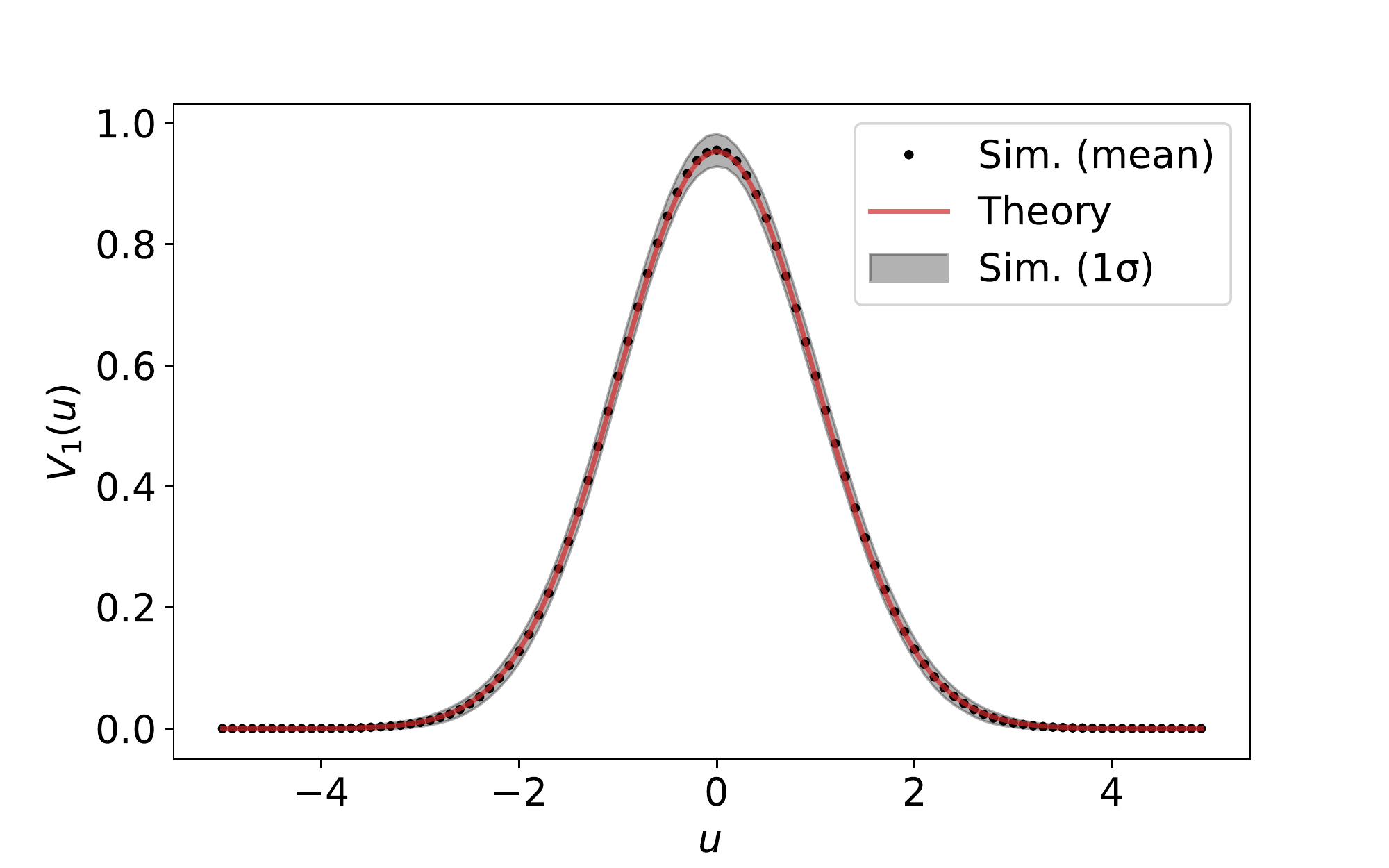}
\end{minipage}
\vspace{-0.6cm}

\begin{minipage}{0.48\textwidth}
    \includegraphics[width=\linewidth]{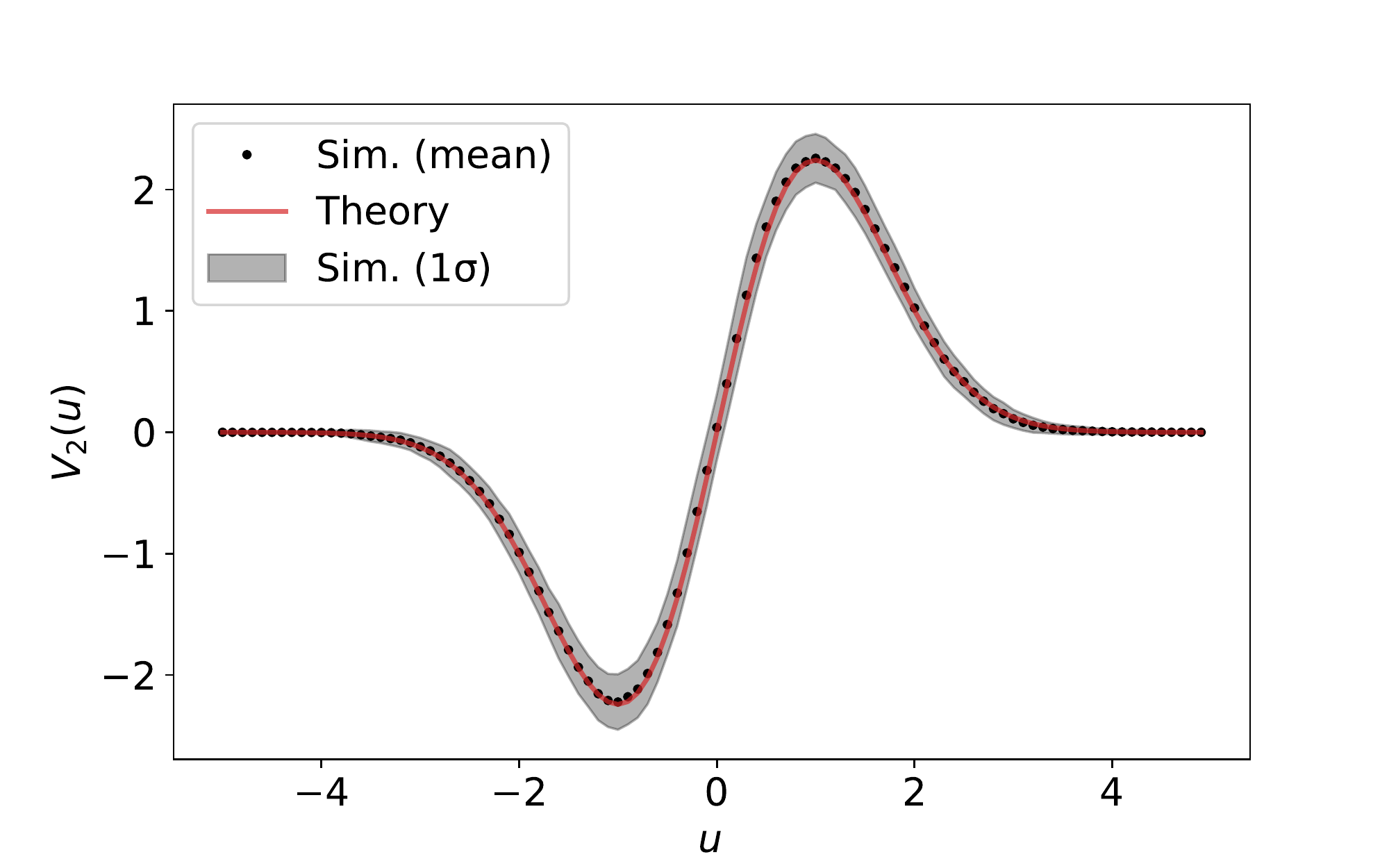}
\end{minipage}
\begin{minipage}{0.48\textwidth}
    \includegraphics[width=1\linewidth]{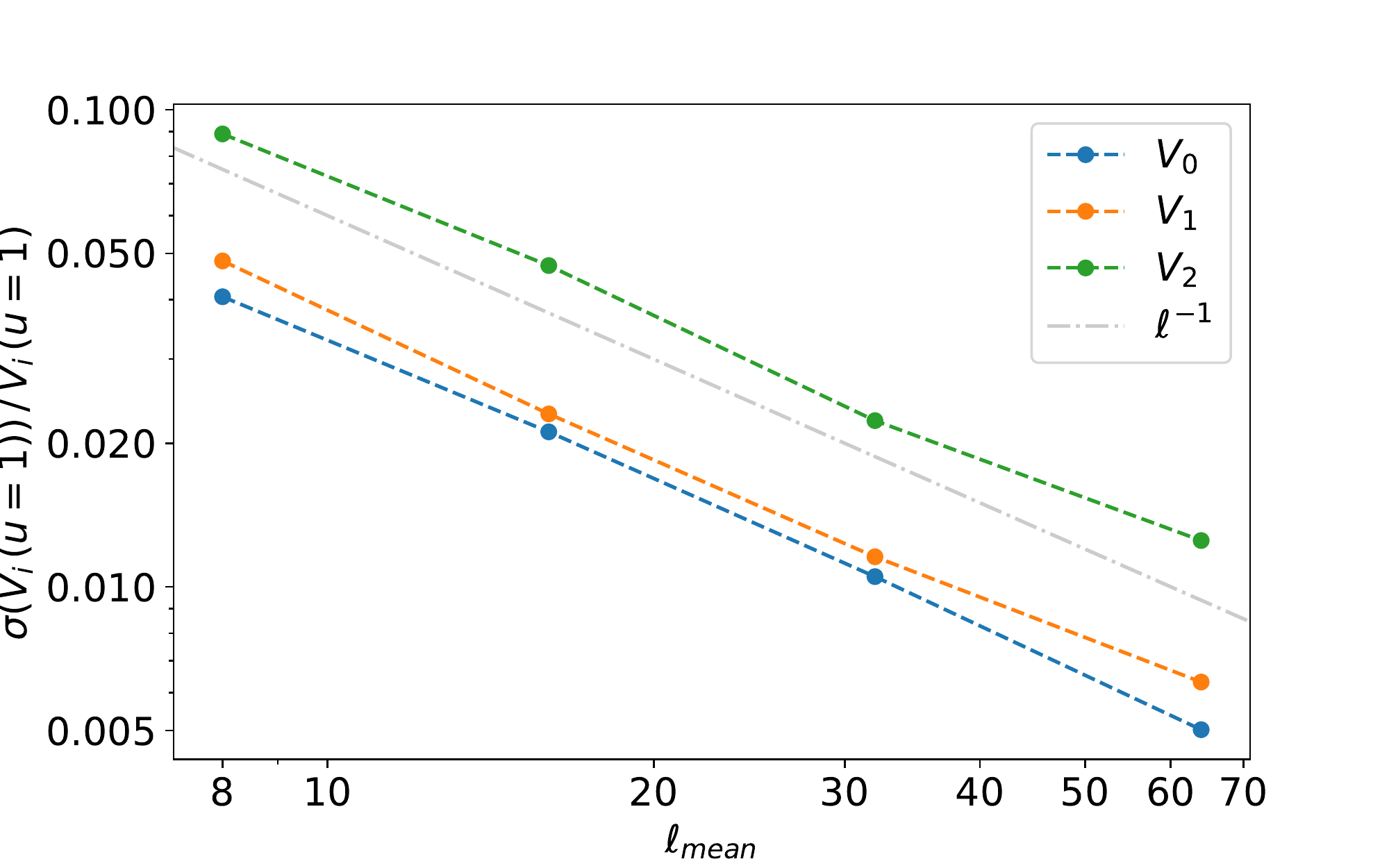}
\end{minipage}
\caption{The results of the three Minkowski Functionals on needlet components of spherical fields. In all cases, the theory is compared with the mean and standard deviation of $100$ Gaussian isotropic simulations. Top left, top right, and bottom left show the first, second, and third Minkowski Functionals as a function of threshold, for a low multipole needlet component ($B=2$, $j=3$, $\ell_{mean}=8$). Bottom right shows the trend of the relative standard deviation of these statistics with respect to the size of the needlet considered, shown at $u=1$ as an example; this is compared to the theoretically expected $\ell^{-1}$ trend.}
\label{f:mink}
\end{figure}

These results prove that, in the high-frequency limit, the Lipschitz-Killing
Curvatures converge to their expected values even on a single realization of
a Gaussian isotropic map, and the fluctuations follow Gaussian behaviour,
thus allowing for standard chi-square testing procedures for
goodness-of-fit. This is yet another form of high-frequency ergodicity, as discussed in the previous section for the empirical distribution of critical points.

The values of the three Minkowski Functionals can be seen in \textbf{Figure \ref{f:mink}}, including the theoretical expectation and both the mean and standard deviation on $100$ simulations. The three Minkoski Functionals are shown for a large scale needlet component ($B=2$, $j=3$, $\ell_{mean}=8$), as well as the evolution of the standard deviation with needlet scale; the trend closely resembles $\ell^{-1} \sim \frac{1}{B^j}$, as expected.

These ideas have been applied in a series of papers, including an analysis on needlet components of Planck maps \citep{planckIS}. They find that the assumption of Gaussianity and isotropy of the foreground-cleaned maps is consistent with the empirical evidence: the Minkowski Functionals evaluated on the Planck CMB data are compatible with the results on realistic simulations at the $2\sigma$ level across different needlet frequencies \citep[see also][for a study on Non--Gaussianity on Planck data]{planckNG}.

Finally, it is worth noting that Minkowski Functionals, as well as other higher-order statistics, are being explored in several spherical fields within Cosmology. This is in order to extract more information from the possible non-Gaussianities of the fields, as this information is unavailable when using only the angular power spectrum. Some applications include Galactic foregrounds \citep{KraPu,martire2023}, maps of the gravitational lensing of the CMB due to general relativity \citep{euclidxxix, grewal2022, zurcher2022}, and the distribution of matter in the Universe, using either the galaxy distribution \citep{appleby2022, liu2023} or observations of the emission of neutral H in the $21$cm line \citep{Spina2021}.

\section{DIRECTIONS FOR FURTHER RESEARCH}

The analysis of CMB temperature maps has now been very deeply explored in
the last 20 years and, hence, a wide variety of tools and techniques are
available for data analysis. The next couple of decades will present more
sophisticated and mathematically challenging issues. Among these, a special
mention must be devoted to CMB polarization data, which will be the object
of several ground-based and satellite experiments, such as the Simons Observatory, ACT, or LiteBIRD,
among others; see \cite{LiteBird} and the references therein.

It is not easy to do justice to the mathematical complexity of polarization random fields at an introductory
level. In a nutshell, the point is that the incoming photons making up CMB oscillate in the plane orthogonal to their travelling direction (as a standard consequence of Maxwell's equations of electromagnetism). Because of these oscillations, they can be understood as drawing unoriented lines on the tangent plane of every point in the celestial sphere. This line bundle can indeed be observed: see \textbf{Figure \ref{f:pol}} for a visualization of CMB polarization (modulus and direction), and \textbf{Figure \ref{f:tempol}} for a visualization of the polarization direction over the temperature map; both figures with data measured by the Planck Satellite.

\begin{figure}
\includegraphics[width=5in]{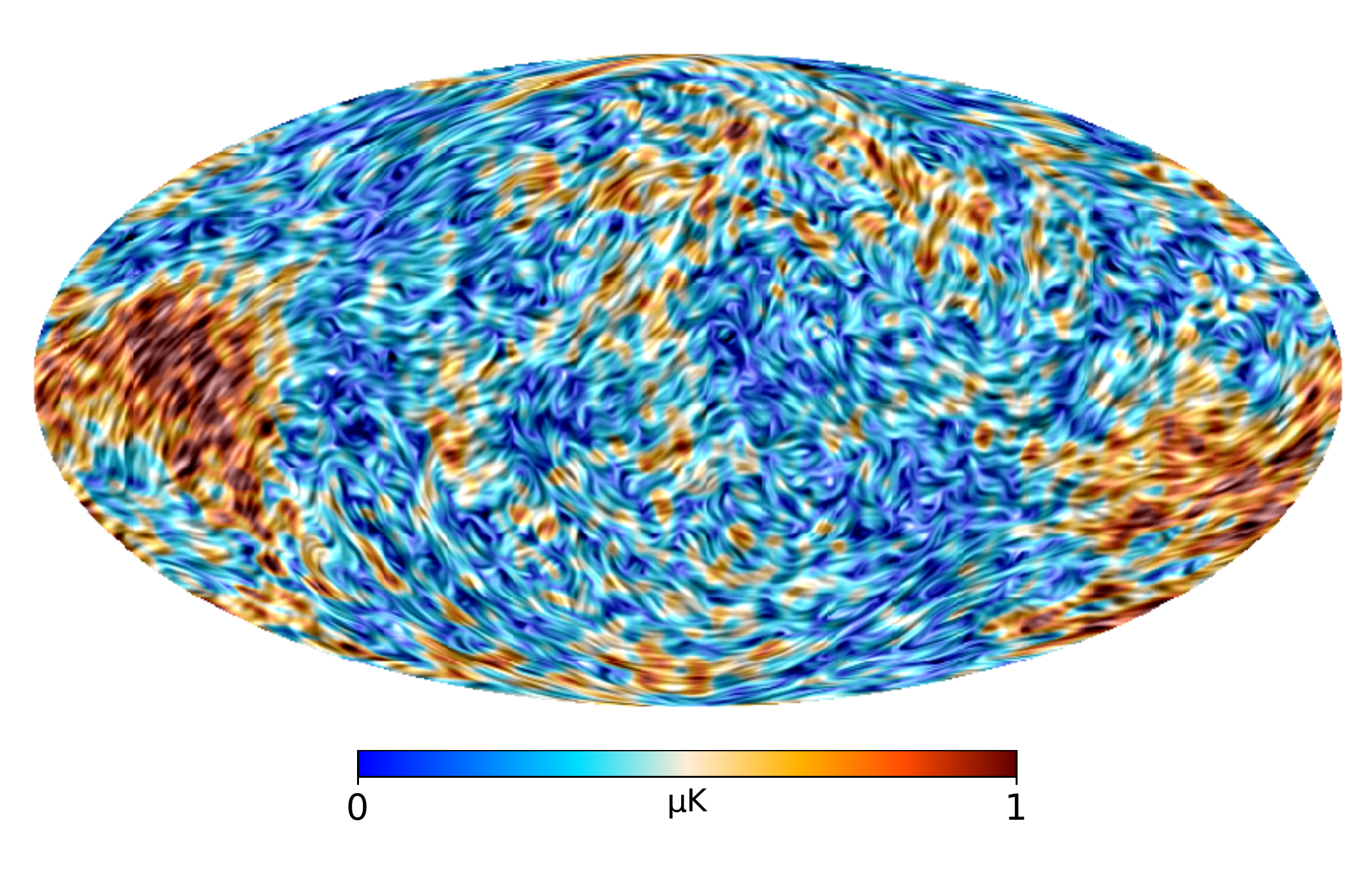}\vspace{-0.5cm}
\caption{CMB polarization, as measured by the Planck satellite. Given that this is a spin 2 field on the sphere, we represent the modulus and direction: the polarization modulus is represented by the background colors, whereas the direction of the polarization is represented through unoriented lines, which constitute the texture of the image. Both modulus and direction have been smoothed with a Gaussian kernel at $fwhm=4^\circ$ for visualization purposes.}
\label{f:pol}
\end{figure}

\begin{figure}
\includegraphics[width=5in]{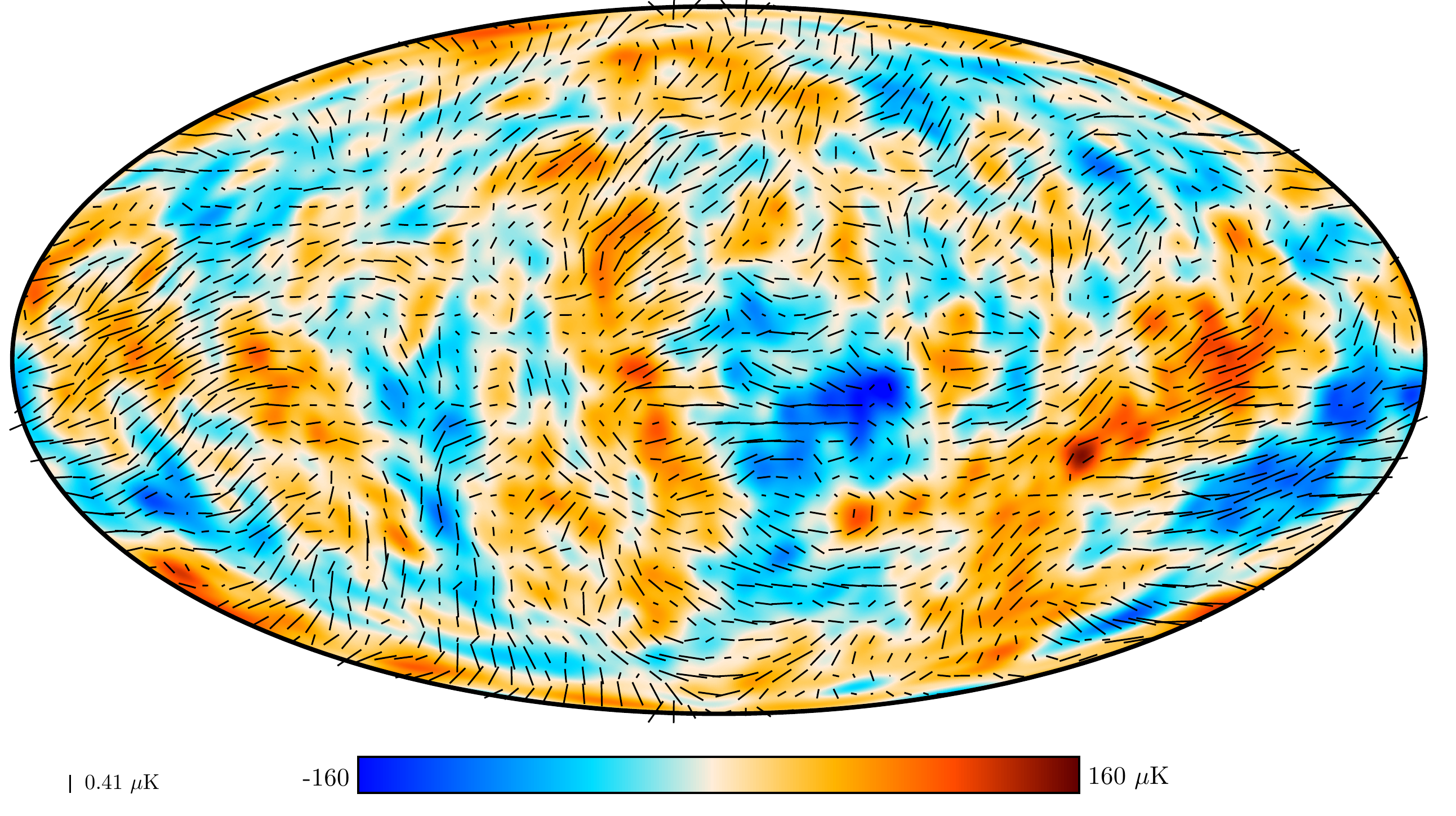}
\caption{CMB temperature and direction of polarization, as measured by the Planck satellite. The color represents the value of the CMB temperature, while the unoriented lines represent the direction of the polarization. Both the temperature and the polarization direction have been smoothed with a Gaussian kernel at $fwhm=5^\circ$ for visualization purposes. Figure reproduced from \cite{planck}; copyright 2020 by The European Southern Observatory.}
\label{f:tempol}
\end{figure}

Mathematically, this can be modeled by considering polarization data as a
realization of a random spin fiber bundle, as discussed for in instance in
\cite{gm2010,leo,maly,BR2013,stecconi2021,LMRS2022}. The idea of spin fiber
bundles was first introduced in a classical paper by \cite{np65}, where it is argued that a quantity $f_{s}(x)$ behaves as a
spin bundle of order $s$ if it transforms as follows under a rotation of $%
\gamma $ radians of the tangent plane at $x:$%
\begin{equation*}
f_{s}^{\prime }(x)=\exp (is\gamma )\,f_{s}(x)\text{ , }\:\gamma \in [0,2\pi )\text{ .}
\end{equation*}%
In the case of line bundles, $s=2$. Indeed, polarization at any point $%
x\in \mathbb{S}^{2}$ is invariant with respect to rotations of $\pi $
radians (\textit{i.e.}, 180 degrees). Spin random fields could also be seen as
tensor-valued, see again \cite{leo,maly} and the references therein.
We note that it is also possible to give a
spectral representation theorem \citep{gm2010,BR2013} for spin random fields, which takes
the form%
\begin{equation}
f_{s}(x)=\sum_{\ell ,m}a_{\ell m;2}Y_{\ell m;2}(x)\text{ , }
\label{srt-spin}
\end{equation}%
where we have introduced the spin spherical harmonics and the spin raising
operators
\begin{equation*}
Y_{\ell m;2}(x)=\partial _{1}\partial_0 Y_{\ell m}\text{ , }
\end{equation*}%
\begin{equation*}
\partial_0 :=(\frac{\partial }{\partial \theta }+\frac{i}{\sin \theta }\frac{%
\partial }{\partial \varphi })\text{ , }\quad\partial _{1}:=-\sin \theta (\frac{%
\partial }{\partial \theta }+\frac{i}{\sin \theta }\frac{\partial }{\partial
\varphi })\frac{1}{\sin \theta }\text{ .}
\end{equation*}

The expansion in Equation \ref{srt-spin} must again be taken with some care, as both the
left and right-hand sides are not invariant to change of local coordinates
in the tangent plane, even if the coordinates for $x\in \mathbb{S}^{2}$
remain the same.

The analysis of polarization data has an enormous importance for Cosmology. In
short, polarization can provide a compelling proof of the existence of primordial
gravitational waves, which would constitute an impressive verification of the
key prediction of Inflation, a theorized epoch of fast expansion in the very Early Universe (the first $10^{-32}$ seconds of the Universe). Other interesting aspects that can be studied with polarization involve the existence of a ``reionization bump'' in the angular power spectrum and the presence of weak gravitational lensing
effects. However, polarization data are much fainter than the CMB temperature,
and map-making and foreground removal are especially challenging, as proved, for instance, by the joint analysis of the Bicep/KECK and Planck team in \cite{BicepPlanck}.

To conclude this paper, we list a number of very challenging tasks,
mostly related to the analysis of polarization data discussed in this section. These are likely to require major statistical efforts in the next ten years.

\begin{enumerate}
\item The implementation of techniques for map-making in polarization. This task is
especially difficult because much less is known about the emission of
astrophysical foreground sources in polarization with respect to temperature
data; some attempts have been made to address these issues by
means of neural network techniques \citep{KraPu}, but the field is still
largely open for research. Very recently, attempts have been made to extend the ideas of NILC and its generalizations to the polarization framework, as in \cite{Car1,Car2}.

\item The development of goodness-of-fit tests for the assumptions of
Gaussianity and isotropy is definitely a very urgent issue, as shown
by the misclassification of Galactic dust emissions in some recent analysis of
polarization data. The implementation on polarization data of
Lipschitz-Killing Curvatures is made difficult by the fact that the
definition of excursion sets is much more subtle in the case of fiber
bundles: a possible approach is to focus on the squared norm of spin data,
which is a scalar quantity following a chi-square distribution with two
degrees of freedom. This is the approach followed by \cite{CarCar2022},
where the expected values of these functionals were also derived, exploiting
results by \cite{LMRS2022}. Much remains to be done to investigate the
statistical properties of such functionals in this framework: for instance,
nothing is currently known about their asymptotic variance nor asymptotic
distributions.

\item An alternative approach to polarization data can be pursued by lifting
the field on the group of rotations $SO(3)$ where it is a scalar-valued
(but anisotropic) field; see \cite{stecconi2021} for a mathematical discussion. In
this case, the machinery by \cite{RFG} to compute the expected values of
Lipschitz-Killing Curvatures becomes much more challenging, because the
fields are no longer isotropic. However, to leading order, the main
results are addressed in \cite{Carron2023}. Computation of variances and
limiting distributions are still completely open for research.

\item A different approach for the derivation of statistics based on
geometric and topological functionals (including Betti numbers) is pursued in \cite%
{LMRS2022}, where the excursion sets are given a much more general
characterization in terms of the behaviour of the field, its gradient and
Hessian, suitably defined for the spin case. Although this approach has made
possible the calculation of some expected values, nothing is currently known
on the variance and distribution of these statistics.

\item The sparsity of foreground components in the needlet domain could be exploited to improve its estimation and removal, as done by \cite{Oppizzi} in CMB temperature. The relevant needlet and wavelet construction is known in polarization \citep{gm2010} but the implementation of thresholding
and other sparsity enforcing techniques is still open for research and applications.

\item Very little is currently known about polarized point sources. The
extension to spin data of the STEM procedure that we discuss in Section \ref{s:ps} for
the scalar case therefore seems a very natural goal. However, a number of tools are
still to be derived, starting from the distribution of local maxima
for spin-valued random fields (in either of the approaches that we envisaged
before, \textit{i.e.}, either considering the corresponding norms as chi-square fields
or searching for maxima in the anisotropic field on $SO(3)$).

\item With a more novel approach, considering that polarization observations
are collected on many different frequency channels, it may be possible to
address the foreground estimation issue in the framework of functional data
analysis on the sphere. This area is completely
open for research, even for the scalar (temperature) case: an attempt to consider statistical analysis for data on the sphere, which take values in a Hilbert space, has been very recently given by \cite{Caponera2022}.
\end{enumerate}

This list of topics is by no means exhaustive, but we hope it will be enough
to motivate some readers to get interested in this challenging and fascinating area of research.



\section*{DISCLOSURE STATEMENT}
The authors are not aware of any affiliations, memberships, funding, or financial holdings that
might be perceived as affecting the objectivity of this review.

\section*{ACKNOWLEDGMENTS}
The research by JC has been supported by the InDark INFN project. The research by DM has been supported by the MUR Department of Excellence Programme MatMotTov.

\appendix
\section{SPHERICAL HARMONICS AND THEIR MAIN PROPERTIES}

In this Appendix, we discuss briefly the formalism of Fourier analysis on
the sphere, and we introduce the main properties of spherical harmonics. Let
us consider first the change of variables into spherical coordinates $%
(x,y,z)=(r\sin \theta \cos \varphi ,r\sin \theta \sin \varphi ,r\cos \theta
)$, where $r^{2}=x^{2}+y^{2}+z^{2}$, $\theta \in \lbrack 0,\pi ]$, $\varphi
\in \lbrack 0,2\pi )$. With this change of coordinates, the 3-dimensional
Laplacian operator $\Delta _{\mathbb{R}^{3}}=\frac{\partial ^{2}}{\partial
x^{2}}+\frac{\partial ^{2}}{\partial y^{2}}+\frac{\partial ^{2}}{\partial
z^{2}}$ takes the form%
\begin{equation*}
\frac{1}{r^{2}}\frac{\partial }{\partial r}r^{2}\frac{\partial }{\partial r}+%
\frac{1}{r^{2}}\Delta _{\mathbb{S}^{2}}\text{ , where }\Delta _{\mathbb{S}%
^{2}}=\frac{1}{\sin \theta }\frac{\partial }{\partial \theta }\sin \theta
\frac{\partial }{\partial \theta }+\frac{1}{\sin ^{2}\theta }\frac{\partial
^{2}}{\partial \varphi ^{2}}\text{ ;}
\end{equation*}%
the operator $\Delta _{\mathbb{S}^{2}}$ is called the spherical Laplacian.
The idea is to decompose the space of square integrable functions on the
sphere, $L^{2}(\mathbb{S}^{2})$, into an orthonormal system of polynomials,
restricted to live on the unit sphere. We first consider homogeneous
polynomials of order $\ell$, \textit{\textit{i.e.}}, linear combinations of terms
of the form $p_{\ell }(x,y,z)=x^{\alpha _{1}}y^{\alpha _{2}}z^{\alpha _{3}}$, with
$\alpha _{1}+\alpha _{2}+\alpha _{3} = \ell$. It is possible to show that the space
of homogeneous polynomials of degree $\ell $ can be decomposed as%
\begin{equation*}
Q_{\ell }(x,y,z)=H_{\ell }(x,y,z)+r^{2}H_{\ell -2}(x,y,z)+r^{4}H_{\ell
-4}(x,y,z)+...
\end{equation*}%
where $H_{\star}(x,y,z)$ denotes the space of harmonic polynomials, \textit{i.e.}, such
that $\Delta _{\mathbb{R}^{3}}h(x,y,z)\equiv 0$ for all $h\in H_{\star}$. Because
we are interested in the restrictions to the sphere, it is hence clearly enough
to focus on polynomials that are both homogeneous and harmonic. Now note
that any homogeneous polynomial can be written in spherical
coordinates as $p_{\ell }(r,\theta ,\varphi )=r^{\ell }Y_{\ell }(\theta
,\varphi )$, leading to%
\begin{eqnarray*}
&&\left\{ \frac{1}{r^{2}}\frac{\partial }{\partial r}r^{2}\frac{\partial }{%
\partial r}+\frac{1}{r^{2}}\Delta _{\mathbb{S}^{2}}\right\} p_{\ell
}(r,\theta ,\varphi ) \\
&=&r^{\ell -2}\ell (\ell +1)Y_{\ell }(\theta ,\varphi )+r^{\ell -2}\Delta _{%
\mathbb{S}^{2}}Y_{\ell }(\theta ,\varphi )=0\text{ .}
\end{eqnarray*}%
In other words, we must search for eigenfunctions of the spherical Laplacian
with eigenvalue $\lambda_\ell = -\ell (\ell +1).$

The following properties of spherical harmonics are important for a proper
understanding of the statistical issues and techniques reviewed in this
paper.

\begin{enumerate}
\item The space of eigenfunctions corresponding to the $\ell $--th
eigenvalue ($\lambda _{\ell }=-\ell (\ell +1)$) has dimensions $2\ell +1.$
Eigenfunctions belonging to different eigenspaces are orthogonal; indeed,
recalling that the Laplacian is a self-adjoint operator (\textit{i.e.}, $\int_{%
\mathbb{S}^{2}}(\Delta _{\mathbb{S}^{2}}f)(x)g(x)dx=\int_{\mathbb{S}%
^{2}}f(x)(\Delta _{\mathbb{S}^{2}}g)(x)dx,$ for all $f,g\in L^{2}(\mathbb{S}%
^{2}))$ we easily have that%
\begin{eqnarray*}
\left\langle Y_{\ell },Y_{\ell ^{\prime }}\right\rangle _{L^{2}(\mathbb{S}%
^{2})} &=&\int_{\mathbb{S}^{2}}Y_{\ell }(x)\overline{Y_{\ell ^{\prime
}}}(x)dx=-\frac{1}{\ell (\ell +1)}\int_{\mathbb{S}^{2}}(\Delta _{\mathbb{S}%
^{2}}Y_{\ell })(x)\overline{Y_{\ell ^{\prime }}}(x)dx \\
&=&-\frac{1}{\ell (\ell +1)}\int_{\mathbb{S}^{2}}Y_{\ell }(x)(\Delta _{%
\mathbb{S}^{2}}\overline{Y_{\ell ^{\prime }}})(x)dx \\
&=&\frac{\ell ^{\prime }(\ell ^{\prime }+1)}{\ell (\ell +1)}\int_{\mathbb{S}%
^{2}}Y_{\ell }(x)\overline{Y_{\ell ^{\prime }}}(x)dx=\frac{\ell ^{\prime
}(\ell ^{\prime }+1)}{\ell (\ell +1)}\left\langle Y_{\ell },Y_{\ell ^{\prime
}}\right\rangle _{L^{2}(\mathbb{S}^{2})}\text{ ,}
\end{eqnarray*}%
which implies the result.

\item A standard orthonormal basis basis for these spaces is formed by the
so-called (complex-valued) Fully Normalized Spherical Harmonics, defined by%
\begin{eqnarray*}
Y_{\ell m}(\theta ,\varphi ) &=&\sqrt{\frac{2\ell +1}{4\pi }}\sqrt{\frac{%
(\ell -m)!}{(\ell +m)!}}P_{\ell m}(\cos \theta )\exp (im\varphi )\text{ ,
for }m\geq 0\text{ ,} \\
Y_{\ell m}(\theta ,\varphi ) &=&(-1)^{m}\overline{Y_{\ell ,-m}}(\theta
,\varphi )\text{ for }m<0\text{ ,}
\end{eqnarray*}%
where $m=-\ell ,...,\ell $ and we have introduced the associated Legendre
functions%
\begin{equation*}
P_{\ell m}(t):=\frac{1}{2^{\ell }\ell !}(1-t^{2})^{m/2}\frac{d^{\ell +m}}{%
dt^{\ell +m}}(t^{2}-1)^{\ell }\text{ ,}
\end{equation*}%
see \cite{AtkHan} and \cite{marpecbook} for more discussion and details.

\item In the special case where $m=0,$ the associated Legendre polynomials
coincide with Legendre polynomials, defined by%
\begin{equation*}
P_{\ell }(t):=\frac{1}{2^{\ell }\ell !}\frac{d^{\ell }}{dt^{\ell }}%
(t^{2}-1)^{\ell }.
\end{equation*}%
Legendre polynomials form an orthogonal base for the space of square
integrable functions from $[-1,1]$ into $\mathbb{R}$.

\item By construction, we have that%
\begin{equation}
\int_{\mathbb{S}^{2}}Y_{\ell m}(x)\overline{Y_{\ell ^{\prime }m^{\prime
}}}(x)dx=\int_{\mathbb{S}^{2}}Y_{\ell m}(\theta ,\varphi )\overline{Y_{\ell
^{\prime }m^{\prime }}}(\theta ,\varphi )\sin \theta d\varphi d\theta =\delta
_{\ell }^{\ell ^{\prime }}\delta _{m}^{m^{\prime }}\text{ .}
\label{orthogonality}
\end{equation}%
The crucial properties of spherical harmonics are given by the so-called
Addition Formula and Duplication Formula. The first holds for any system
of orthogonal elements and states the following: given $x_{1},x_{2}\in
\mathbb{S}^{2}$, we have that%
\begin{equation}
\sum_{m=-\ell }^{\ell }Y_{\ell m}(x_{1})\overline{Y_{\ell m}}(x_{2})=\frac{%
2\ell +1}{4\pi }P_{\ell }(\left\langle x_{1},x_{2}\right\rangle )\text{ ,}
\label{addition}
\end{equation}%
see \cite{AtkHan}, Theorem 2.9, or \cite{marpecbook}, Section 3.4.2, for a
proof. The duplication formula, on the other hand, can be stated as follows:
for any $x_{1},x_{2}\in \mathbb{S}^{2},$ we have that%
\begin{equation}
\int_{\mathbb{S}^{2}}\frac{2\ell +1}{4\pi }P_{\ell }(\left\langle
x_{1},x\right\rangle )\frac{2\ell +1}{4\pi }P_{\ell }(\left\langle
x,x_{2}\right\rangle )dx=\frac{2\ell +1}{4\pi }P_{\ell }(\left\langle
x_{1},x_{2}\right\rangle )\text{ .}  \label{duplication}
\end{equation}%
The proof follows immediately by using eq \ref{addition} twice and eq \ref%
{orthogonality}.
\end{enumerate}

The overwhelming majority of statistical results does not require any
explicit manipulation of the analytic expressions for the $Y_{\ell m}(.);$
these expressions are numerically implemented on packages such as
HealPix \citep{healpix}. On the other hand, Equations \ref{orthogonality}, \ref%
{addition}, and \ref{duplication} are the basis for nearly every theoretical
argument in the analysis of spherical random fields.

\section{THE SPECTRAL REPRESENTATION THEOREM}

The Spectral Representation Theorem plays such a crucial role in the
analysis of isotropic spherical random fields that we feel it is useful to
give a proper statement and a sketch of its proof here. Let us first recall that a function $\Gamma (.,.):\mathbb{S}^{2}\times \mathbb{S}^{2}\rightarrow
\mathbb{R}$ is said to be non-negative if and only if for all $p\in \mathbb{N%
}$, $x_{1},...,x_{p}\in \mathbb{S}^{2}$, and $\alpha _{1},...,\alpha _{p}\in
\mathbb{R}$ we have that%
\begin{equation*}
\sum_{j,k=1}^{p}\alpha _{j}\alpha _{k}\Gamma (x_{j},x_{k})\geq 0\text{ .}
\end{equation*}%
Given that we are concerned with isotropic fields, we have that $\Gamma
(x_{j},x_{k})=\Gamma (x_{j}^{\prime },x_{k}^{\prime })$ whenever $%
\left\langle x_{j},x_{k}\right\rangle =\left\langle x_{j}^{\prime
},x_{k}^{\prime }\right\rangle$; in these cases we say that the function $\Gamma
(.,.)$ is isotropic and, with some abuse of notation, we write $\Gamma
(x_{j},x_{k})=\Gamma (\left\langle x_{j},x_{k}\right\rangle )$. The next
ingredient that we have to consider is Schoenberg's Theorem, which reads as follows:

\begin{theorem}[Schoenberg]
Let $\Gamma (\left\langle x_{1},x_{2}\right\rangle ):%
\mathbb{S}^{2}\times \mathbb{S}^{2}\rightarrow \mathbb{R}$ be an isotropic,
continuous, non-negative definite function. Then there exist a sequence of
non-negative weights $\left\{ C_{\ell }\right\} _{\ell =0,1,...}$ such that $%
\sum_{\ell =0}^{\infty }\frac{2\ell +1}{4\pi }C_{\ell }=\Gamma (\left\langle
x,x\right\rangle )<\infty $ and
\begin{equation*}
\Gamma (\left\langle x_{1},x_{2}\right\rangle )=\sum_{\ell =0}^{\infty }%
\frac{2\ell +1}{4\pi }C_{\ell }P_{\ell }(\left\langle
x_{1},x_{2}\right\rangle )\text{ ,}
\end{equation*}%
uniformly for all $x_{1},x_{2}\in \mathbb{S}^{2}$.
\end{theorem}

Therefore, the family of non-negative isotropic functions on $\mathbb{S}^{2}\times
\mathbb{S}^{2}$ can be identified with the covariance functions of isotropic random fields. Schoenberg's Theorem is hence a generalization of classical
results by Herglotz and Bochner: it states that the autocovariance
functions can be expressed as an inverse Fourier transform of the \emph{%
angular power spectrum} $\left\{ C_{\ell }\right\} _{\ell =0,1,...}$. The
result requires the autocovariance function to be uniformly continuous;
however, it was proved by \cite{mp2012} that this is necessarily the case if
the isotropic spherical random field is measurable.

Having in mind this result, the Spectral Representation Theorem can be stated as follows:
\begin{theorem}[Spectral Representation]
Let $T:\Omega \times \mathbb{S}^{2}\rightarrow \mathbb{R}$ be an isotropic
random field on the sphere with finite variance. Then there exists a family
of orthogonal random coefficients $\left\{ a_{\ell m}\right\} _{m=-\ell
,...,\ell;\, \ell =0,1,2...}$ such that $\mathbb{E}\left[ a_{\ell m}\overline{a%
_{\ell ^{\prime }m^{\prime }}}\right] =\delta _{\ell }^{\ell ^{\prime
}}\delta _{m}^{m^{\prime }}C_{\ell }$ and
\begin{equation*}
T(x,\omega )=\sum_{\ell =0}^{\infty }\sum_{m=-\ell }^{\ell }a_{\ell
m}(\omega )Y_{\ell m}(x)\text{ ,}
\end{equation*}%
where the equality holds both in the $L^{2}(\Omega )$ sense for fixed $x\in
\mathbb{S}^{2}$ and in $L^{2}(\Omega \times \mathbb{S}^{2}).$
\end{theorem}

\begin{proof}
(Sketch) A sketch of the proof can be given as follows. Denote by $\mathcal{%
H}=\overline{span}\left\{ T(x),\text{ }x\in \mathbb{S}^{2}\right\} \subset
L^{2}(\Omega )$ the closure of the space generated by the linear
combinations of $T(x),$ and define the linear isometry%
\begin{equation*}
D:\mathcal{H\rightarrow }L^{2}(\mathbb{S}^{2})\text{ , }D(T(x)):=\sum_{\ell
=0}^{\infty }\frac{2\ell +1}{4\pi }\sqrt{C_{\ell }}P_{\ell }(\left\langle
.,x\right\rangle )\text{ .}
\end{equation*}%
Finally, let us call $\mathcal{K=}\overline{span}\left\{ \sum_{\ell
=0}^{\infty }\frac{2\ell +1}{4\pi }C_{\ell }P_{\ell }(\left\langle
.,x\right\rangle ),\text{ }x\in \mathbb{S}^{2}\right\} \subset L^{2}(\mathbb{%
S}^{2}).$ It is readily seen that this application is indeed an isometry, in
fact by Schoemberg's Theorem we have that%
\begin{eqnarray*}
\mathbb{E}\left[ T(x_{1})T(x_{2})\right] &=&\left\langle
T(x_{1}),T(x_{2})\right\rangle _{L^{2}(\Omega )}=\sum_{\ell =0}^{\infty }%
\frac{2\ell +1}{4\pi }C_{\ell }P_{\ell }(\left\langle
x_{1},x_{2}\right\rangle ) \\
&=&\int_{\mathbb{S}^{2}}\sum_{\ell =0}^{\infty }\frac{2\ell +1}{4\pi }\sqrt{%
C_{\ell }}P_{\ell }(\left\langle x_{1},x\right\rangle )\sum_{\ell ^{\prime
}=0}^{\infty }\frac{2\ell ^{\prime }+1}{4\pi }\sqrt{C_{\ell ^{\prime }}}%
P_{\ell ^{\prime }}(\left\langle x,x_{2}\right\rangle )dx \\
&=&\left\langle D(T(x_{1})),D(T(x_{2}))\right\rangle _{L^{2}(\mathbb{S}^{2})}%
\text{ ,}
\end{eqnarray*}%
using the duplication formula and the definition of the inner product in $L^{2}(\mathbb{S}^{2})$. It is also
readily seen, for instance by approximating the integral with Riemann sums,
that we have%
\begin{eqnarray*}
a_{\ell m}(\omega ) &\coloneqq&\int_{\mathbb{S}^{2}}T(x,\omega )\overline{Y_{\ell
m}}(x)dx\:\in \mathcal{H}\text{ ,} \\
D(a_{\ell m}(\omega )) &=&\int_{\mathbb{S}^{2}}D(T(x,\omega ))\overline{Y%
_{\ell m}}(x)dx =\int_{\mathbb{S}^{2}}\sum_{\ell =0}^{\infty }\frac{2\ell +1}{4\pi }%
C_{\ell }P_{\ell }(\left\langle .,x\right\rangle )\overline{Y_{\ell m}}(x)dx=%
\sqrt{C_{\ell }}\,\overline{Y_{\ell m}}(.)\text{ ,}
\end{eqnarray*}%
using again \ref{addition} and \ref{orthogonality}; note that%
\begin{equation*}
\mathbb{E}\left[ a_{\ell m}\overline{a_{\ell ^{\prime }m^{\prime }}}\right]
=\left\langle a_{\ell m},a_{\ell ^{\prime }m^{\prime }}\right\rangle
_{L^{2}(\Omega )}=\left\langle \sqrt{C_{\ell }}\,\overline{Y_{\ell m}}(.),%
\sqrt{C_{\ell ^{\prime }}}\,\overline{Y_{\ell ^{\prime }m^{\prime
}}}(.)\right\rangle _{L^{2}(\Omega )}=\delta _{\ell }^{\ell ^{\prime }}\delta
_{m}^{m^{\prime }}C_{\ell }\text{ .}
\end{equation*}%
Moreover, the application $D$ is an isometry and thus injective and
invertible, so that%
\begin{eqnarray*}
T(x,\omega ) &=&D^{-1}\left(\sum_{\ell =0}^{\infty }\frac{2\ell +1}{4\pi }\sqrt{%
C_{\ell }}P_{\ell }(\left\langle .,x\right\rangle )\right) \\
&=&\sum_{\ell =0}^{\infty }\frac{2\ell +1}{4\pi }D^{-1}\left(\sqrt{C_{\ell }}\,%
\overline{Y_{\ell m}}(.)\right)Y_{\ell m}(x) \\
&=&\sum_{\ell =0}^{\infty }\sum_{m=-\ell }^{\ell }a_{\ell m}(\omega )Y_{\ell
m}(x)\text{ ,}
\end{eqnarray*}%
where we have used, once again, linearity and Equation \ref{addition}; the proof is
then completed.
\end{proof}

As a final remark, it follows from the proof that the spherical harmonic coefficients are
always uncorrelated, and hence independent, in the Gaussian case.
Surprisingly, the reverse turns out to be the case as well: for isotropic
random fields, if the spherical harmonic coefficients are independent, then
they are necessarily Gaussian, see \cite{BM-SPL}. Some ongoing research is devoted to provide further characterizations of the spherical harmonics coefficients.


\begin{thebibliography}{00}
\bibitem[Adler \& Taylor(2007)]{RFG}
{ Adler RJ, Taylor JE. 2007. \emph{Random Fields and Geometry}, Springer. }

\bibitem[Adler \& Taylor(2011)]{AT2}
Adler RJ, Taylor JE. 2011. \emph{Topological
complexity of smooth random functions}. Lecture Notes in Mathematics, 2019

\bibitem[Atkinson \& Han(2012)]{AtkHan}
Atkinson K, Han W. 2012. \emph{Spherical harmonics and
approximations on the unit sphere: an introduction.} Lecture Notes in
Mathematics, 2044. Springer.

\bibitem[Appleby et~al(2022)]{appleby2022}
{Appleby S, Park C, Pranav P, Hong SE, Hwang HS, Kim J, Buchert T. 2022. Minkowski Functionals of SDSS-III BOSS: Hints of Possible Anisotropy in the Density Field? \emph{The Astrophysical Journal}, 928(2):108.}

\bibitem[Axelsson et~al(2015)]{axelsson}
{Axelsson M, Ihle HT, Scodeller S, Hansen FK. 2015. Testing for foreground residuals in the Planck foreground cleaned maps: A new method for designing confidence masks, \emph{%
Astronomy and Astrophysics}, 578:A44. }

\bibitem[Baldi et~al(2009a)]{bkmpAoS}
{Baldi P, Kerkyacharian G, Marinucci D, Picard D. 2009a. Asymptotics for Spherical Needlets, \emph{Annals of
Statistics}, 37:1150-71. }

\bibitem[Baldi et~al(2009b)]{BKMP}
{Baldi P, Kerkyacharian G, Marinucci D, Picard D. 2009b. Subsampling needlet coefficients on the sphere, \emph{%
Bernoulli}, 15:438--63. }

\bibitem[Baldi \& Marinucci(2007)]{BM-SPL}
Baldi P, Marinucci D. 2007. Some characterizations of the spherical harmonics coefficients for isotropic random fields. \emph{%
Statist. Probab. Lett.} 77(5):490--96.

\bibitem[Baldi \& Rossi(2014)]{BR2013}
Baldi P, Rossi M. 2014. Representation of Gaussian isotropic spin random fields. \emph{Stochastic Process. Appl.} 124(5):1910--41.

\bibitem[Banjamini \& Hochberg(1995)]{BH}
{Benjamini Y, Hochberg Y. 1995. Controlling the false discovery rate: a practical and powerful approach to multiple testing.
\emph{J. Roy. Statist. Soc. Ser. B}, 57:289--300. }

\bibitem[Paoletti et~al(2022)]{BicepPlanck}
Paoletti D, Finelli F, Valiviita J, Hazumi M. 2022. Planck and BICEP/Keck Array 2018 constraints on primordial gravitational waves and perspectives for future B -mode polarization measurements, \emph{Physical Review D}, 106(8),083528

\bibitem[Bobin et~al(2014)]{starck2014}
{Bobin J, Sureau F, Starck JL, Rassat A, Paykari P. 2014. Joint Planck and WMAP CMB Map Reconstruction,
\emph{Astronomy and Astrophysics}, 563:A105. }

\bibitem[Brockwell \& Davis(2006)]{BD}
Brockwell PJ, Davis RA. 2006. \emph{Time series: theory
and methods.} Reprint of the second (1991) edition. Springer Series in
Statistics. Springer, New York.

\bibitem[Cammarota \& Marinucci(2015)]{CM2015}
Cammarota V, Marinucci D. 2015. The stochastic properties of $\ell ^{1}$-regularized spherical Gaussian fields. \emph{Appl.
Comput. Harmon. Anal.} 38(2):262--83.

\bibitem[Cammarota et~al(2016)]{cmw2014}
{Cammarota V, Marinucci D, Wigman I. 2016. On the distribution of the critical values of random spherical harmonics, \emph{Journal of Geometric Analysis}, 26:3252--324. }

\bibitem[Caponera(2022)]{Caponera2022}
Caponera A. 2022. Asymptotics for isotropic Hilbert-valued spherical random fields. arXiv:2212.02329

\bibitem[Carones et~al(2022a)]{CarCar2022}
Carones A, Carr\'{o}n Duque J, Marinucci D, Migliaccio M, Vittorio N. 2022a. Minkowski Functionals of CMB polarisation intensity with Pynkowski: theory and application to Planck data. arXiv:2211.07562.

\bibitem[Carones et~al.(2022b)Carones, Migliaccio, Marinucci \& Vittorio]{Car1}
Carones A, Migliaccio M, Marinucci D, Vittorio N. 2022b. Analysis of NILC performance on B-modes data of sub-orbital experiments. arXiv:2208.12059.

\bibitem[Carones et~al(2022c)]{Car2}
Carones A, Migliaccio M, Puglisi G, Baccigalupi C, Marinucci D, Vittorio N, Poletti D. 2022c. Multi-Clustering Needlet-ILC for CMB B-modes component separation. arXiv:2212.04456.

\bibitem[Carr\'on Duque et~al(2019)]{CarronDuque}
Carr\'on Duque J, Buzzelli A, Fantaye Y, Marinucci D, Schwartzman A, Vittorio N. 2019. Point Source Detection and False Discovery Rate Control on CMB Maps, \emph{Astronomy and Computing}, 28:100310

\bibitem[Carr\'on Duque et~al(2023)]{Carron2023}
Carr\'on Duque J, Carones A, Marinucci D, Migliaccio M, Vittorio N. 2023. Minkowski Functionals in $SO(3)$ for spin-2 CMB polarisation fields. arXiv:2301.13191.

\bibitem[Cheng et~al(2020)]{ChengCammarota}
Cheng D, Cammarota V, Fantaye Y, Marinucci D, Schwartzman A. 2020. Multiple testing of local maxima for detection of peaks on the (celestial) sphere. \emph{Bernoulli} 26(1):31--60.

\bibitem[Cheng \& Schwartzman(2015)]{chengschwartzman1}
{Cheng D, Schwartzman A. 2015. Distribution of the height of local maxima of Gaussian random fields. \emph{%
Extremes}, 18:213--40. }

\bibitem[Cheng \& Schwartzman(2017)]{chengschwartzman2}
Cheng D, Schwartzman A. 2017. Multiple
testing of local maxima for detection of peaks in random fields. \emph{Ann.
Statist.} 45(2):529--56.

\bibitem[Cheng \& Schwartzman(2018)]{chengschwartzman2015m}
Cheng D, Schwartzman A. 2018. Expected
number and height distribution of critical points of smooth isotropic
Gaussian random fields. \emph{Bernoulli} 24(4B):3422--46.

\bibitem[Dodelson(2003)]{dode2004}
{Dodelson S. 2003. \emph{Modern Cosmology}, Academic Press.}

\bibitem[Durrer(2008)]{Durrer}
{Durrer R. 2008. \emph{The Cosmic Microwave Background}, Cambridge University Press. }

\bibitem[Euclid Collaboration(2023)]{euclidxxix}
{Euclid Collaboration. 2023. Euclid Preparation XXIX: Forecasts for 10 different higher-order weak lensing statistics. arXiv:2301.12890.}

\bibitem[Geller \& Mayeli(2009a)]{gm1}
{Geller D, Mayeli A. 2009a. Continuous Wavelets on
Compact Manifolds, \emph{Math. Z.,} 262:895--927. }

\bibitem[Geller \& Mayeli(2009b)]{gm2}
{Geller D, Mayeli A. 2009b. Nearly Tight Frames
and Space-Frequency Analysis on Compact Manifolds, \emph{Math. Z.,} 263:235--64. }

\bibitem[Geller \& Mayeli(2009c)]{gm3}
{Geller D, Mayeli A. 2009c. Besov Spaces and
Frames on Compact Manifolds, \emph{Indiana Univ. Math. J.,} 58:2003--42. }

\bibitem[Geller \& Marinucci(2010)]{gm2010}
Geller D, Marinucci D. 2010. Spin wavelets on the
sphere. \emph{J. Fourier Anal. Appl.} 16(6):840--84.

\bibitem[G{\'{o}}rski et~al(2005)]{healpix}
{G{\'{o}}rski KM, Hivon E, Banday AJ, Wandelt BD, Hansen FK, Reinecke M, Bartelmann M. 2005. HEALPix: A Framework for High-Resolution Discretization and Fast Analysis of
Data Distributed on the Sphere, \emph{Astrophysical Journal}, 699:759--71.}

\bibitem[Grewal et~al(2022)]{grewal2022}
{Grewal N, Zuntz J, Tröster T, Amon A. 2022. Minkowski Functionals in Joint Galaxy Clustering \& Weak Lensing Analyses, \emph{The Open Journal of Astrophysics}, 5.}

\bibitem[Hazumi et~al(2020)]{LiteBird}
{Hazumi M, Ade PAR, Adler A, Allys E, Arnold K, et~al. 2022. LiteBIRD satellite: JAXA's new strategic L-class mission for all-sky surveys of cosmic microwave background polarization. \emph{Proceedings of the SPIE}, 11443:114432F.}

\bibitem[Krachmalnicoff \& Puglisi(2021)]{KraPu}
Krachmalnicoff N, Puglisi G. 2021. ForSE: a GAN based algorithm for extending CMB foreground models to sub-degree angular scales,
\emph{The Astrophysical Journal.} 911:42.

\bibitem[Liu et~al(2023)]{liu2023}
{Liu W, Jiang A, Fang W. 2023. Probing massive neutrinos with the Minkowski functionals of the galaxy distribution, arXiv:2302.08162.}

\bibitem[Leonenko \& Sakhno(2012)]{leo}
Leonenko N, Sakhno L. 2012. On spectral representations of
tensor random fields on the sphere. \emph{Stoch. Anal. Appl.} 30(1):44--66.

\bibitem[Lerario et~al(2022)]{LMRS2022}
Lerario A, Marinucci D, Rossi M, Stecconi M. 2022. Geometry and topology of spin random fields. arXiv:2207.08413

\bibitem[Loh(2005)]{loh}
{Loh WL. 2005. Fixed-Domain Asymptotics for a
Subclass of Mat\'{e}rn-type Gaussian Random Fields, \emph{Annals of
Statistics}, 33:2344--94. }

\bibitem[Loh(2015)]{loh2}
{Loh WL. 2015. Estimating the smoothness of a
Gaussian random field from irregularly spaced data via higher-order
quadratic variations, \emph{Annals of Statistics}, 43:2766--94. }

\bibitem[Malyarenko(2013)]{maly}
Malyarenko A. 2013. \emph{Invariant random fields on spaces
with a group action.} Probability and its Applications (New York). Springer,
Heidelberg

\bibitem[Marinucci \& Peccati(2011)]{marpecbook}
{Marinucci D, Peccati G. 2011. \emph{Random
Fields on the Sphere. Representation, Limit Theorem and Cosmological
Applications}, Cambridge University Press }

\bibitem[Marinucci \& Peccati(2013)]{mp2012}
{Marinucci D, Peccati G. 2013. Mean Square
Continuity on Homogeneous Spaces of Compact Groups, \emph{Electronic
Communications in Probability}. 18(37):1--10.}

\bibitem[Marinucci et~al(2008)]{mpbb08}
{Marinucci D, Pietrobon D, Balbi A, Baldi P,
Cabella P, et~al. 
2008. Spherical Needlets for CMB Data Analysis, \emph{Monthly Notices of
the Royal Astronomical Society}, 383:539--45. }

\bibitem[Martire et~al(2023)]{martire2023}
{Martire FA, Banday AJ, Martínez-González E, Barreiro RB. 2023. Morphological Analysis of the Polarized Synchrotron Emission with WMAP and Planck. arXiv:2301.08041.}

\bibitem[Narcowich et~al(2006a)]{npw1}
{Narcowich FJ, Petrushev P, Ward JD. 2006a. Localized Tight Frames on Spheres, \emph{SIAM Journal of Mathematical
Analysis }, 38:574--94. }

\bibitem[Narcowich et~al(2006b)]{npw2}
{Narcowich FJ, Petrushev P, Ward JD. 2006b.
Decomposition of Besov and Triebel-Lizorkin Spaces on the Sphere, \emph{%
Journal of Functional Analysis,} 238:530--64. }

\bibitem[Newman \& Penrose(1966)]{np65}
Newman ET, Penrose R. 1966. Note on the
Bondi-Metzner-Sachs group. \emph{J. Mathematical Phys.} 7:863--70.

\bibitem[Oppizzi et~al(2020)]{Oppizzi}
Oppizzi F, Renzi A, Liguori M, Hansen FK, Baccigalupi C, et~al.%
2020. Needlet Thresholding Methods in Component Separation, \emph{Journal of Cosmology and Astroparticle Physics}. 3:054.

\bibitem[Planck Collaboration(2015)]{planckCS}
{Planck Collaboration. 2015. Planck 2015. XXVI.
The second Planck catalogue of compact sources, \emph{Astronomy and Astrophysics}. 594:A26. }

\bibitem[Planck Collaboration(2020a)]{planck}
{Planck Collaboration. 2020a. Planck 2018 results I. Overview, and the cosmological legacy of Planck, \emph{Astronomy and Astrophysics}. 641:A1. }

\bibitem[Planck Collaboration(2020b)]{planckMM}
{Planck Collaboration. 2020b. Planck 2018 results IV. Diffuse component separation, \emph{Astronomy and Astrophysics}. 641:A4. }

\bibitem[Planck Collaboration(2020c)]{planckPS}
{Planck Collaboration. 2020c. Planck 2018 results V. CMB power spectra and likelihoods, \emph{Astronomy and Astrophysics}. 641:A5. }

\bibitem[Planck Collaboration(2020d)]{planckCosmo}
{Planck Collaboration. 2020d. Planck 2018 results VI. Cosmological parameters, \emph{Astronomy and Astrophysics}. 641:A6. }

\bibitem[Planck Collaboration(2020e)]{planckIS}
{Planck Collaboration. 2020e. Planck 2018 results VII. Isotropy and Statistics of the CMB, \emph{Astronomy and Astrophysics}. 641:A7. }

\bibitem[Planck Collaboration(2020f)]{planckNG}
{Planck Collaboration. 2020f. Planck 2018 results IX. Constraints on primordial non-Gaussianity, \emph{Astronomy and Astrophysics}. 641:A9. }

\bibitem[Remazeilles et~al(2011)]{Delabrouille}
Remazeilles M, Delabrouille J, Cardoso JF. 2011. Foreground component separation with generalized Internal Linear
Combination, \emph{Monthly Notices of the Royal Astronomical Society, }%
418(1):467--76.

\bibitem[Scodeller et~al(2011)]{scodellermexican}
{Scodeller S, Rudjord O, Hansen FK, Marinucci D, Geller D, Mayeli A. 2011. Introducing Mexican needlets for CMB analysis: Issues for practical applications and comparison with
standard needlets, \emph{Astrophysical Journal}. 733:121. }

\bibitem[Scodeller et~al(2012)]{scodeller}
{Scodeller S, Hansen FK, Marinucci D. 2012. Detection of new point sources in WMAP 7 year data using internal templates and needlets, \emph{Astrophysical Journal}, 753:27.}

\bibitem[Scodeller \& Hansen(2012)]{scodeller2}
{Scodeller S, Hansen FK. 2012. Masking
versus removing point sources in CMB data: the source corrected WMAP power
spectrum from new extended catalogue, \emph{Astrophysical Journal}, 761:119. }

\bibitem[Schwartzman et~al(2011)]{Schwartzman:2011}
{Schwartzman A, Gavrilov Y, Adler RJ. 2011. Multiple testing of local maxima for detection of peaks in 1D,
\emph{Annals of Statistics}, 39:3290--319. }

\bibitem[Shevchenko \& Todino(2023)]{SheTo}
Shevchenko R, Todino AP. 2023. Asymptotic behaviour of
level sets of needlet random fields. \emph{Stochastic Process. Appl.} 155:268--318.

\bibitem[Spina et~al(2021)]{Spina2021}
{Spina B, Porciani C, Schimd C. 2021. The HI-halo mass relation at redshift $z \sim 1$ from the Minkowski functionals of 21-cm intensity maps. \emph{Monthly Notices of the Royal Astronomical Society}, 505(3):3492-504}

\bibitem[Stecconi(2022)]{stecconi2021}
Stecconi M. 2022. Isotropic random spin weighted
functions on $\mathbb{S}^{2}$ vs isotropic random fields on $\mathbb{S}^{3}$%
. \emph{Theory Probab. Math. Statist.} 107:77--109.

\bibitem[Vittorio(2018)]{Vittorio}
Vittorio, N. (2018) \emph{Cosmology, }CRC Press.

\bibitem[Zürcher et~al(2021)]{zurcher2022}
{Zurcher D, Fluri J, Sgier R, Kacprzak T, Refregier A. 2021. Cosmological forecasts for non-Gaussian statistics in Large-Scale Weak Lensing Surveys. \emph{Journal of Cosmology and Astroparticle Physics}, 2021(01):028}


\end{thebibliography}
\end{document}